 \documentclass[final,3p,times]{elsarticle}

\usepackage{amsfonts}
\usepackage{graphicx}
\usepackage{times}
\usepackage{color}
\usepackage{amssymb}
\usepackage{natbib}
\usepackage{amsmath}
\usepackage{bm}
\usepackage{bbold}
\usepackage[bbgreekl]{mathbbol}
\usepackage{breqn}
\usepackage{algorithm}
\usepackage{algpseudocode}
\usepackage{accents}
\usepackage{multicol}
\usepackage{subcaption}
\usepackage{qcircuit}  

\font\ppppppcarac=ptmr8y at 4pt
\font\pppppcarac=ptmr8y at 5pt
\font\ppppcarac=ptmr8y at 6pt
\font\pppcarac=ptmr8y at 7pt

\font\pcarac=ptmr8y at 9pt

\newcommand{\bfF}{{\bm{F}}}

\newcommand{\bfH}{{\bm{H}}}

\newcommand{\bfQ}{{\bm{Q}}}

\newcommand{\bfW}{{\bm{W}}}
\newcommand{\bfX}{{\bm{X}}}
\newcommand{\bfY}{{\bm{Y}}}

\newcommand{\bfzero}{{ \hbox{\bf 0} }}

\newcommand{\bfb}{{\bm{b}}}

\newcommand{\bfg}{{\bm{g}}}

\newcommand{\bfk}{{\bm{k}}}

\newcommand{\bfm}{{\bm{m}}}

\newcommand{\bfw}{{\bm{w}}}
\newcommand{\bfx}{{\bm{x}}}
\newcommand{\bfy}{{\bm{y}}}

\newcommand{\bfalpha}{{\bm{\alpha}}}
\newcommand{\bfbeta}{{\bm{\beta}}}

\newcommand{\bfeta}{{\bm{\eta}}}

\newcommand{\bftau}{{\bm{\tau}}}

\newcommand{\CC}{{\mathbb{C}}}

\newcommand{\HH}{{\mathbb{H}}}

\newcommand{\MM}{{\mathbb{M}}}
\newcommand{\NN}{{\mathbb{N}}}

\newcommand{\RR}{{\mathbb{R}}}

\newcommand{\ggeclair}{{\mathbb{g}}}

\DeclareMathAlphabet{\mathonebb}{U}{bbold}{m}{n}
\def\11{{\ensuremath{\mathonebb{1}}}}

\newcommand{\curB}{{\mathcal{B}}}
\newcommand{\curD}{{\mathcal{D}}}
\newcommand{\curE}{{\mathcal{E}}}

\newcommand{\curM}{{\mathcal{M}}}

\newcommand{\curP}{{\mathcal{P}}}

\newcommand{\curR}{{\mathcal{R}}}

\newcommand{\curT}{{\mathcal{T}}}

\newcommand{\curV}{{\mathcal{V}}}

\newcommand{\bfcurY}{{\boldsymbol{\mathcal{Y}}}}

\newcommand{\DM}{{\hbox{{\pppppcarac DM}}}}

\newcommand{\Div}{{\hbox{\pcarac div}}}

\newcommand{\exactp}{{{\hbox{{\ppppcarac ref}}}}}

\newcommand{\FKP}{{\hbox{{\ppppcarac FKP}}}}

\newcommand{\optp}{{\hbox{{\ppppcarac opt}}}}
\newcommand{\optpp}{{\hbox{{\pppppcarac opt}}}}

\newcommand{\TB}{{\hbox{{\pppppcarac TB}}}}

\newcommand{\training}{{\hbox{{\pppcarac train}}}}

\newcommand{\VQE}{{\hbox{{\pppppcarac VQE}}}}
\newcommand{\pVQE}{{\hbox{{\ppppppcarac VQE}}}}

\usepackage{mathrsfs}
\newcommand{\curC}{{\mathscr{C}}}       
\font\bf=ptmb8y at 10pt

\font\vcarac=ptmr8y at 10pt

\newcommand{\vc}{{\hbox{\vcarac{v}}}}   

\newdefinition{definition}{Definition}

\newtheorem{proposition}{Proposition}

\newproof{proof}{Proof}
\newdefinition{remark}{Remark}
\newdefinition{hypothesis}{Hypothesis}
\newdefinition{notation}{Notation}
\newproof{example}{Example}
\numberwithin{equation}{section}

\journal{ArXiv}

\begin{document}

\begin{frontmatter}

\title{Quantum computer formulation of the FKP-operator eigenvalue problem for probabilistic learning on manifolds}


\author[1]{Christian Soize \corref{cor1}}
\ead{christian.soize@univ-eiffel.fr}
\author[1]{Lo\"ic Joubert-Doriol}
\ead{ loic.joubert-doriol@univ-eiffel.fr}
\author[2]{Artur F.  Izmaylov}
\ead{artur.izmaylov@utoronto.ca}
\cortext[cor1]{Corresponding author: C. Soize, christian.soize@univ-eiffel.fr}
\address[1]{Universit\'e Gustave Eiffel, MSME UMR 8208, 5 bd Descartes, 77454 Marne-la-Vall\'ee, France}
\address[2]{University of Toronto  Scarborough, Department of Physical and Environmental Sciences, 1265 Military Trail Toronto, Canada}

\begin{abstract}
We present a quantum computing formulation to address a challenging problem in the development of probabilistic learning on manifolds (PLoM). It involves solving the spectral problem of the high-dimensional Fokker-Planck (FKP) operator, which remains beyond the reach of classical computing. Our ultimate goal is to develop an efficient approach for practical computations on quantum computers. For now, we focus on an adapted formulation tailored to quantum computing.
The methodological aspects covered in this work include the construction of the FKP equation, where the invariant probability measure is derived from a training dataset, and the formulation of the eigenvalue problem for the FKP operator. The eigen equation is transformed into a Schr\"odinger equation with a potential $\curV$, a non-algebraic function that is neither simple nor a polynomial representation. To address this, we propose a methodology for constructing a multivariate polynomial approximation of $\curV$, leveraging polynomial chaos expansion within the Gaussian Sobolev space. This approach preserves the algebraic properties of the potential and adapts it for quantum algorithms.
The quantum computing formulation employs a finite basis representation, incorporating second quantization with creation and annihilation operators. Explicit formulas for the Laplacian and potential are derived and mapped onto qubits using Pauli matrix expressions. Additionally, we outline the design of quantum circuits and the implementation of measurements to construct and observe specific quantum states. Information is extracted through quantum measurements, with eigenstates constructed and overlap measurements evaluated using universal quantum gates.
\end{abstract}
\end{frontmatter}

\section{Introduction}
\label{Section1}
%
We present a quantum computing methodology to address a challenging problem recently encountered in the development of probabilistic learning on manifolds \cite{Soize2024b}. This approach aims to construct statistical surrogate models from small training datasets, which are particularly useful for uncertainty quantification, nonconvex optimization under uncertainty, and statistical inverse problems. The main challenge lies in solving the spectral problem associated with the high-dimensional Fokker-Planck (FKP) operator, a task currently beyond the capabilities of classical computing.
In this paper, we explore a quantum computing-based methodology for calculating the eigenfunctions associated with the smallest eigenvalues of the FKP operator. The ultimate goal is to develop an efficient approach that, in the near future, could enable effective computation on a quantum computer. For now, we focus on an adapted formulation tailored to quantum computing.
It is worth noting that this problem may also be relevant to frameworks and applications beyond those that motivated the developments presented in this work.

%
\paragraph{Quantum Computing Applications in Computational Mechanics and Engineering}
Since research on quantum information and quantum computation began in the 1980s, numerous papers have been published, covering all aspects necessary for developing quantum computing techniques to address large-scale simulations in engineering science, including data science. This effort requires the development of suitable quantum computing algorithms, software advancements, quantum error correction methods, and the definition of hardware and infrastructure requirements for executing large-scale numerical simulations.
These challenges are particularly complex and demand interdisciplinary collaboration among quantum physicists, mathematicians, computer scientists, and engineers. Given the vast literature in this multidisciplinary field, this short introduction does not aim to review all existing works. However, the engineering science community has recently shown interest in leveraging quantum computers for numerical simulations and algorithm development. For example, quantum computing applications have been explored in computational mechanics \cite{LiuB2024,Wulff2024,Xu2024}, materials science \cite{Ma2020,Micheletti2021,Mizuta2021}, finite element methods \cite{Raisuddin2022}, and engineering simulation and optimization \cite{WangYan2023,Ye2023}.

%
\paragraph{Objectives and organization of the paper}
The aim of this paper is to propose a method to solve the spectral problem of the high-dimensional Fokker-Planck (FKP) operator, which remains beyond the capabilities of classical computing.  It is known (see, for example, \cite{Risken1989,Gardiner1985}) that the FKP operator can be reformulated as a Schr\"odinger-type operator exhibiting a potential $\curV$.
This reformulation bears a strong resemblance to the vibrational spectral problem in molecular systems, where quantum computing has recently gained traction as an appealing alternative to classical computers~\cite{McArdle2019,Sawaya2019}. We follow a similar path in the current paper, also aligning with recent efforts to implement the dynamical solution of the FKP problem on quantum computers~\cite{Amaro2024}.
The FKP equation under consideration is constructed within the framework of probabilistic learning on manifolds (PLoM) \cite{Soize2024b}, as we will explain in Section~\ref{Section2}. The potential $\curV$, associated with a probability density function related to the training dataset, is a non-algebraic function that is neither simple nor a polynomial representation, and therefore does not satisfy the algebraic prerequisites of the quantum algorithm we will consider.
In the first part of this paper, we present in Sections~\ref{Section3} and \ref{Section4} the construction of a representation of the potential of the Schr\"odinger operator. To achieve this, we propose a methodology for constructing a multivariate polynomial approximation \cite{McArdle2019,Sawaya2019} of $\curV$, leveraging polynomial chaos expansion within the Gaussian Sobolev space \cite{Malliavin1995}. This approach preserves the algebraic properties of the potential and adapts it for quantum algorithms. It enables the construction of a polynomial chaos expansion in Gaussian space with good convergence properties, preserving the algebraic characteristics that control the spectrum of the operator. Section~\ref{Section5} addresses the numerical validation of the computation of the polynomial chaos expansion coefficients of the potential.
The second part of the paper, Sections~\ref{Section6} to \ref{Section8}, focuses on the quantum computing formulation. This formulation employs a finite basis representation (spectral method) that incorporates second quantization with creation and annihilation operators. Explicit formulas for the Laplacian and potential are derived and mapped onto qubits using Pauli matrix expressions. The design of quantum circuits and the implementation of measurements to construct and observe specific quantum states are presented in Section~\ref{Section7}. Finally, Section~\ref{Section8} addresses the construction of eigenstates and the extraction of overlaps through quantum measurements, evaluated using universal quantum gates.
%
\paragraph{Key steps in the proposed methodology}
Hereinafter, we present methodological aspects. Although presented in the context of the FKP-operator eigenvalue problems, the proposed methodology illustrates well the main steps required to use quantum computers.\

\noindent (i) Construction of the FKP equation whose invariant probability measure, associated with the steady-state solution, is derived from a given training dataset used within the framework of probabilistic learning on manifolds.\

\noindent (ii) Formulation of the eigenvalue problem related to the FKP operator and transforming it into a Schr\"odinger-type operator exhibiting a potential $\bfy\mapsto\curV(\bfy)$ on $\RR^\nu$ with high dimension $\nu$.\

\noindent (iii) Construction of a polynomial chaos expansion (PCE) of the potential $\curV$ in the Gaussian Sobolev space \cite{Malliavin1995}, and introduction of a convergence analysis for the truncated PCE representation of $\curV$. Transformation of the truncated PCE into a polynomial in $\bfy^\bfm$, in which $\bfm$ is a multi-index of dimension $\nu$.
This is the final representation of $\curV$ that is perfectly adapted for developing a quantum algorithm based on a second quantization.\

\noindent (iv) Development of a quantum computing formulation based on the following three steps: (a) Defining a finite basis representation, for which the eigenvalue problem can be rewritten in a second quantized form. (b) Developing the second quantization, introducing the creation and annihilation operators in order to rewrite the eigenvalue problem in the Fock space, allowing us to obtain explicit formulas for the Laplacian operator $\nabla^2$ and for the multivariate monomials $\bfy^\bfm$. Deducing a representation of the Fokker-Planck operator as a linear combination with known real coefficients of products of powers of the creation and annihilation operators. (c) Mapping onto a system of qubits by expressing the products of creation and annihilation operators with the Pauli matrices.\

\noindent (v) Design of the quantum circuits and measurements in order to construct and measure specific states on the quantum computer. (a) The quantum circuits are constructed with gates. To construct the eigenstates, a set of universal quantum gates is used from which any unitary transformation can be reconstructed. The successive application of these gates on the set of qubits forms the quantum circuits. (b) The extraction of information from the quantum states is performed via quantum measurements of certain physical observables, which are probabilistic.\

\noindent (vi) Finally, we perform the eigenstate construction and overlap extraction.
%
\subsection{Convention for the variables, vectors, and matrices}
\label{Section1.4}
\noindent $x,\eta$: lower-case Latin or Greek letters are deterministic real variables.\\
$\bfx,\bfeta$: boldface lower-case Latin or Greek letters are deterministic vectors.\\
$X$: upper-case Latin letters are real-valued random variables.\\
$\bfX$: boldface upper-case Latin letters are vector-valued random variables.\\
$[x]$: lower-case Latin letters between brackets are deterministic matrices.\\
$[\bfX]$: boldface upper-case letters between brackets are matrix-valued random variables.
%
%
\subsection{Convention used for random variables}
\label{Section1.6}
In this paper, for any finite integer $m\geq 1$, the Euclidean space $\RR^m$ is equipped with the $\sigma$-algebra $\curB_{\RR^m}$. If $\bfY$ is a $\RR^m$-valued random variable defined on the probability space $(\Theta,\curT,\curP)$, $\bfY$ is a  mapping $\theta\mapsto\bfY(\theta)$ from $\Theta$ into $\RR^m$, measurable from $(\Theta,\curT)$ into $(\RR^m,\curB_{\RR^m})$, and $\bfY(\theta)$ is a realization (sample) of $\bfY$ for $\theta\in\Theta$. The probability distribution of $\bfY$ is the probability measure $P_\bfY(d\bfy)$ on the measurable set $(\RR^m,\curB_{\RR^m})$ (we will simply say on $\RR^m$). The Lebesgue measure on $\RR^m$ is noted $d\bfy$ and
when $P_\bfY(d\bfy)$ is written as $p_\bfY(\bfy)\, d\bfy$, $p_\bfY$ is the probability density function (pdf) on $\RR^m$ of $P_\bfY(d\bfy)$ with respect to $d\bfy$. Finally, $E$ denotes the mathematical expectation operator that is such that
$E\{\bfY\} = \int_{\RR^m} \bfy \, P_\bfY(d\bfy)$.
%
\section{Probabilistic learning on manifolds with a given probability measure}
\label{Section2}
%
\subsection{Preamble related to the context of the probabilistic learning on manifolds}
\label{Section2.1}
Probabilistic Learning on Manifolds (PLoM) is a tool in computational statistics, introduced in 2016 \cite{Soize2016}, which can be viewed as a tool for scientific machine learning. The PLoM approach has specifically been developed for small training-dataset cases \cite{Soize2016,Soize2020c}. The method avoids the scattering of learned realizations associated with the probability distribution to preserve its concentration in the neighborhood of the random manifold defined by the parameterized computational model.
Several extensions have been proposed to account for implicit constraints induced by physics, computational models, and measurements \cite{Soize2020a,Soize2021a,Soize2022b,Soize2024}, to reduce the stochastic dimension using a statistical partition approach \cite{Soize2022a}, and to update the prior probability distribution with a target dataset, whose points are, for instance, experimental realizations of the system observations \cite{Soize2022c}. Consequently, PLoM, constrained by a stochastic computational model and statistical moments or samples/realizations, allows for performing probabilistic learning inference and constructing predictive statistical surrogate models for large parameterized stochastic computational models.\\

PLoM approach relies on projecting, along the data axis, a matrix-valued Itô equation linked to a stochastic dissipative Hamiltonian system, serving as the MCMC generator for the probability measure estimated via Gaussian Kernel Density Estimation (GKDE) on the training dataset. The projection basis is the diffusion maps (DMAPS) basis, associated with a time-independent isotropic kernel introduced in \cite{Coifman2005,Coifman2006}, referred to as the reduced-order DMAPS basis in the PLoM context.
Since 2016, all published PLoM applications  have used the isotropic kernel, yielding quality results even for heterogeneous data and complex systems in various dimensions. Improving statistical surrogates for stochastic manifolds with conditional statistics via PLoM using a projection basis from a transient anisotropic (time-dependent) kernel has recently been proposed \cite{Soize2024b}.\\

The first algorithmic step of PLoM \cite{Soize2016,Soize2020c} involves performing a principal component analysis (PCA) on the non-Gaussian random vector $\bfX$, which is generally in high dimension. This vector $\bfX$ results from concatenating the random vector $\bfQ$, representing the quantities of interest (QoI), with the random vector $\bfW$, representing the random control variables of the considered stochastic system. Additionally, this system depends on random latent/uncontrolled variables. The vector-valued random QoI, $\bfQ = \bfF(\bfW)$, is connected to $\bfW$ by an unknown random nonlinear mapping. Denoting by $\curC_w$ the support of the probability measure of $\bfW$, the random graph  $\{(\bfw,\bfF(\bfw))\, \vert \, \bfw\in\curC_w\}$ forms a stochastic manifold due to the latent random variables in the system implying the randomness of mapping $\bfF$. The PCA is performed on the training dataset of $\bfX$. The new coordinates resulting from the PCA form a random vector $\bfH$ with a reduced dimension. The training dataset of $\bfH$ is obtained through the PCA projection of the training dataset of $\bfX$. By construction, $\bfH$ is a non-Gaussian normalized random vector, which is centered and has a covariance matrix equal to the identity matrix.
Under these conditions, the available information is only the training dataset of $\bfH$. Using this dataset, the probability density function $p_\bfH$ of $\bfH$ is constructed by using the  modification \cite{Soize2015} of the GKDE method \cite{Bowman1997}. Thus, the only available information is represented by the probability density function $p_\bfH$ of $\bfH$. Given the normalization of $\bfH$, the challenges in constructing the learned dataset of $\bfH$, with or without constraints, involve preserving the concentration of the learned probability measure for $\bfH$ (induced by the random manifold) and ensuring the ``good quality'' of the learned joint probability measure of the components of $\bfH$, not just the marginal probability measure of each component. The quality of the statistical surrogate model of $\bfQ$ given $\bfW$, based on conditional statistics,  is directly linked to the learned joint probability measure of the components of $\bfH$. Here, to simplify the presentation, we will directly begin on $\bfH$ rather than starting from $\bfX = (\bfQ, \bfW)$.\\

The method presented in \cite{Soize2024b} starts from the FKP equation, whose stationary solution is $p_\bfH$. The spectral problem of this FKP operator is introduced, and its  eigenfunctions would form the ideal basis for performing the data projection necessary to preserve the concentration of the learned probability measure. Given the high dimensionality and small training datasets, a change of scale (smoothing) is introduced, similar to DMAPS, to calculate the eigenvectors of a matrix that approximates the FKP operator. This approximation is essential because computing the eigenfunctions of the FKP operator in high dimension  is beyond the capabilities of classical computers. In this paper, we revisit this spectral problem and construct a formulation adapted to quantum computing, potentially enabling the direct calculation of eigenfunctions in high dimensions.
\subsection{Defining the probability measure of random vector $\bfH$}
\label{Section2.2}
Let $\curD_\training(\bfeta) = \{\bfeta^1,\ldots,\bfeta^{n_d}\}$ be the training dataset of $n_d > 1$ independent realizations $\bfeta^j\in\RR^\nu$, with $\nu\geq 1$, of a second-order $\RR^\nu$-valued random variable $\bfH$, defined on a probability space $(\Theta,\curT,\curP)$ and which results of the PCA of random vector $\bfX$ (see Section~\ref{Section2.1}. Let $\underline\bfeta_d\in\RR^\nu$ and $[C_d]\in\MM_{n_d}$ be the associated empirical estimates of the mean value and the covariance matrix constructed with the points of $\curD_\training(\bfeta)$,
\begin{equation} \label{eq2.1}
\underline\bfeta_d = \frac{1}{n_d} \sum_{j=1}^{n_d} \bfeta^j \quad , \quad
[C_d] =  \frac{1}{n_d -1} \sum_{j=1}^{n_d} (\bfeta^j - \underline\bfeta_d )\otimes (\bfeta^j - \underline\bfeta_d )\, .
\end{equation}
The properties of the PCA imply that
\begin{equation} \label{eq2.2}
\underline\bfeta_d =  \bfzero_\nu \quad , \quad [C_d] =  [I_\nu]\, .
\end{equation}
Let $\bfy\mapsto p_\bfH(\bfy)$ be the probability density function on $\RR^\nu$, with respect do the Lebesgue measure $d\bfy$, defined by
\begin{equation} \label{eq2.3}
p_\bfH(\bfy) = \frac{1}{n_d}\sum_{j=1}^{n_d} \frac{1}{(\sqrt{2 \pi} \, \hat s)^\nu} \exp\left ( -\frac{1}{2 \hat s^2} \Vert \bfy - \frac{\hat s}{s}\bfeta^j\Vert^2\right ) \quad , \quad \forall \bfy\in\RR^\nu \, ,
\end{equation}
where $\hat s$ and $s$ are defined by
\begin{equation} \label{eq2.4}
s = \left ( \frac{4}{n_d(2+\nu)}\right )^{1/(\nu+4)} \quad , \quad
\hat s  = \frac{s}{\sqrt{s^2 + (n_d-1)/n_d}} \, .
\end{equation}
Eqs.~\eqref{eq2.3} and \eqref{eq2.4} correspond to the Gaussian kernel-density estimation (KDE) constructed using the $n_d$ independent realizations of $\curD_\training(\bfeta)$ involving the modification \cite{Soize2015} of the usual formulation \cite{Bowman1997,Duong2008,Gentle2009,Givens2013}, in which $s$ is the Silverman bandwidth.
Let $\bfH$ be the second-order $\RR^\nu$-valued random variable, defined on a probability space $(\Theta,\curT,\curP)$, whose probability measure $P_\bfH(d\bfy) = p_\bfH(\bfy)\, d\bfy$ on $\RR^\nu$ is defined by the probability density function $p_\bfH$ given by Eq.~\eqref{eq2.3}. It can be seen that, for any fixed $n_d > 1$, we have exactly,
\begin{equation} \label{eq2.5}
E\{\bfH\} = \int_{\RR^\nu} \bfy \, P_\bfH(d\bfy) = \frac{1}{2\hat s^2} \,\underline\bfeta_d = \bfzero_\nu \, ,
\end{equation}
\begin{equation} \label{eq2.6}
E\{\bfH\otimes\bfH\} = \int_{\RR^\nu} \bfy\otimes\bfy \, P_\bfH(d\bfy)
= \hat s^2\, [I_\nu] + \frac{\hat s^2}{s^2} \, \frac{(n_d-1)}{n_d}\,[C_d] = [I_\nu] \, .
\end{equation}
Eqs.~\eqref{eq2.5} and \eqref{eq2.6} show that $\bfH$ is a normalized $\RR^\nu$-valued random variable. The probability density function $p_\bfH$ defined by Eq.~\eqref{eq2.3} is rewritten, for all $\bfy$ in $\RR^\nu$, as
\begin{equation} \label{eq2.7}
p_\bfH(\bfy) = c_\nu \, \xi(\bfy)  \quad , \quad \xi(\bfy) = e^{-\Phi(\bfy)} \, ,
\end{equation}
in which  $c_\nu =  (\sqrt{2 \pi} \, \hat s)^{-\nu}$ and where $\Phi(\bfy) = -\log(\xi(\bfy))$ is such that
\begin{equation} \label{eq2.8}
\Phi(\bfy) = - \log\left ( \frac{1}{n_d}\sum_{j=1}^{n_d} \exp\left ( -\frac{1}{2 \hat s^2} \Vert \bfy - \frac{\hat s}{s}\bfeta^j\Vert^2\right ) \right ) \quad , \quad \forall \bfy\in\RR^\nu \, .
\end{equation}
%
%
\section{igenvalue problem of the FKP operator and  a representation  adapted to quantum computing algorithm}
\label{Section3}
%
This section summarizes essential results concerning the Fokker-Planck (FKP) operator and its associated eigenvalue problem, which are formally presented. Then, the FKP operator is transformed to obtain a Schr\"odinger-type formulation in which $\curV(\bfy)$ will be the potential on $\RR^\nu$. Finally, we propose a methodology to construct a representation of the potential $\curV(\bfy)$ adapted to quantum computing algorithms, which is presented in Section~\ref{Section6}. This representation will be based on a polynomial chaos expansion (PCE)in Gaussian space of $\curV(\bfy)$, which will be explicitly constructed in Section~\ref{Section4}.

\subsection{It\^o stochastic differential equation related to probability density function $p_\bfH$}
\label{Section3.1}

We introduce an  It\^o stochastic differential equation (ISDE) on $\RR^\nu$, with initial condition, for which $P_\bfH(d\bfy) =p_\bfH(\bfy)\, d\bfy$ is the invariant measure. A classical candidate to such an ISDE is written as
\begin{align}
d\bfcurY(t) = & \, \bfb(\bfcurY(t))\, dt + d\bfW(t) \quad , \quad t > 0 \, , \label{eq3.1}\\
\bfcurY(0) = & \, \bfx \in \RR^\nu \, , \, a.s. \, ,  \label{eq3.2}
\end{align}
where the drift vector is the function $\bfy\mapsto \bfb(\bfy)$ from $\RR^\nu$ into $\RR^\nu$ defined by
\begin{equation} \label{eq3.3}
\bfb(\bfy) = -\frac{1}{2}\, \nabla \Phi(\bfy) \quad , \quad \forall \bfy \in \RR^\nu \, .
\end{equation}
In Eq.~\eqref{eq3.1}, $\bfW(t) = (W_1(t),\ldots ,W_\nu(t))$ is the normalized Wiener stochastic process \cite{Doob1953} on $\RR^+ = [0\, , +\infty[$, with values in $\RR^\nu$, which is a stochastic process with independent increments, such that $\bfW(0) = \bfzero_\nu \, a.s $, and for  $0\leq \tau < t < +\infty$, the increment $\Delta\bfW_{\tau  t} = \bfW(t) - \bfW(\tau)$ is a Gaussian $\RR^\nu$-valued second-order random variable, centered and with a covariance matrix that is written as
\begin{equation} \label{eq3.4}
[C_{\Delta\bfW_{\tau t}}] = E\{\Delta\bfW_{\tau t}\otimes \Delta\bfW_{\tau t}  \} = (t-\tau)\, [I_\nu] \, .
\end{equation}
We will see that $\{\bfY(t), t\in\RR^+\}$ is a homogeneous diffusion stochastic process, which is asymptotically stationary for $t\rightarrow +\infty$. Assuming that the transition probability measure of $\bfY(t)$ given $\bfY(0) = \bfx$ admits a density with respect to d$\bfy$, such that, for all $t > 0$, for all $\bfx$ and $\bfy$ in $\RR^\nu$, and for any Borelian $\curB$ in $\RR^\nu$, we have
\begin{align}
 \curP(\{\bfcurY(t) \in\curB \, \vert \, \bfcurY(0) = \bfx\}  & =  \int_\curB \rho(\bfy,t \, \vert \, \bfx,0) \, d\bfy \, ,  \label{eq3.6}\\
 \lim_{t\rightarrow 0_+}  \rho(\bfy,t \, \vert \, \bfx,0) \, d\bfy &  = \delta_0(\bfy - \bfx)  \, ,  \label{eq3.7}\\
\int_{\RR^\nu} \rho(\bfy,t \, \vert \, \bfx,0) \, d\bfy  & = 1 \, .  \label{eq3.7bis}
\end{align}
%

\subsection{FKP equation associated with the ISDE, and its steady-state solution}
\label{Section3.2}

For all $\bfx$ in $\RR^\nu$, the transition probability density function $(\bfy,t)\mapsto \rho(\bfy,t \, \vert \, \bfx,0)$ from
$\RR^\nu\times\RR^+$ into $\RR^+$ verifies the following Fokker-Planck (FKP) equation (see for instance \cite{Guikhman1980,Friedman2006,Soize1994}),
\begin{equation} \label{eq3.8}
\frac{\partial\rho}{\partial t} + L_\FKP(\rho) = 0 \quad , \quad t > 0 \, ,
\end{equation}
with the initial condition for $t=0$ defined by Eq.~\eqref{eq3.7}. The Fokker-Planck operator $L_\FKP$ can be written, after a small algebraic manipulation and for any sufficiently differentiable function $\bfy\mapsto \vc(\bfy)$ from $\RR^\nu$ into $\RR$, as
\begin{equation} \label{eq3.9}
 \{L_\FKP(\vc)\}(\bfy) = -\frac{1}{2} \Div \left \{ p_\bfH(\bfy) \nabla \left ( \frac{\vc(\bfy)}{p_\bfH(\bfy)}\right ) \right\}  \, .
\end{equation}
The detailed balance (the probability current vanishes) is satisfied and the steady-state solution of Eq.~\eqref{eq3.8} is the pdf $p_\bfH$ defined by Eq.~\eqref{eq2.3} \cite{Soize1994,Gardiner1985,Risken1989}. We then have, for $\vc = p_\bfH$,
\begin{equation} \label{eq3.10}
 L_\FKP(p_\bfH) = 0 \, .
\end{equation}
The invariant measure  $P_\bfH(d\bfy) = p_\bfH(\bfy)\, d\bfy$ is such that, for all $\bfy$ in $\RR^\nu$ and for all $t\geq 0$,
\begin{equation} \label{eq3.10bis}
 p_\bfH(\bfy) = \int_{\RR^\nu} \rho(\bfy,t \, \vert \, \bfx,0) \, p_\bfH(\bfx) \, d\bfx \, .
\end{equation}
The ISDE defined by Eqs.~\eqref{eq3.1} and \eqref{eq3.2} admits an asymptotic ($t\rightarrow +\infty$) stationary solution whose marginal probability density function of order one is $p_\bfH$. Consequently, for all $\bfx$ and $\bfy$ in $\RR^\nu$, we have
\begin{equation} \label{eq3.11}
\lim_{t\rightarrow +\infty}  \rho(\bfy,t \, \vert \, \bfx,0)  = p_\bfH(\bfy) \, .
\end{equation}
%

\subsection{Formal formulation of the eigenvalue problem of the FKP operator}
\label{Section3.3}

The eigenvalue problem, posed in an adapted functional space, is written as
\begin{equation} \label{eq3.12}
 L_\FKP(\vc) =  \lambda\, \vc \, ,
\end{equation}
for which the current must vanish at infinity, yielding the condition,
\begin{equation} \label{eq3.13}
\lim_{\Vert\bfy\Vert \rightarrow +\infty}  p_\bfH(\bfy) \,\Vert\,  \nabla  ( p_\bfH(\bfy)^{-1}\vc(\bfy)   ) \, \Vert \, = 0\, .
\end{equation}
Continuing the development within  a formal framework, such as that used in \cite{Risken1989}, we introduce the change of function,
\begin{equation} \label{eq3.14}
\vc(\bfy) = p_\bfH(\bfy)^{1/2}\, q(\bfy) \quad , \quad \bfy\in\RR^\nu \quad , \quad q: \RR^\nu\rightarrow \RR \, .
\end{equation}
Let $\hat L_\FKP$ be the linear operator defined, for $q$ belonging to an admissible set of functions,
\begin{equation} \label{eq3.15}
\{\hat L_\FKP(q)\}(\bfy) =  p_\bfH(\bfy)^{-1/2}\,L_\FKP(p_\bfH(\bfy)^{1/2}\, q(y)) \quad , \quad \bfy\in\RR^\nu \, .
\end{equation}
Therefore, the eigenvalue problem defined by Eqs.~\eqref{eq3.12} and \eqref{eq3.13} can  be rewritten in $q$ as
\begin{equation} \label{eq3.16}
 \hat L_\FKP(q) =  \lambda\, q \, ,
\end{equation}
with the condition at infinity,
\begin{equation} \label{eq3.17}
\lim_{\Vert\bfy\Vert \rightarrow +\infty} p_\bfH(\bfy) \, \Vert \, \nabla  ( p_\bfH(\bfy)^{-1/2} q(\bfy)   ) \, \Vert \, = 0\, .
\end{equation}
%

\subsection{Schr\"{o}dinger-type formulation of the FKP operator $\hat L_\FKP$}
\label{Section3.4}
Using Eq.~\eqref{eq2.7}, which shows that $p_\bfH(\bfy)^{-1} \nabla_\bfH(\bfy) = - \nabla\Phi$, and using
 Eqs.~\eqref{eq3.9} and \eqref{eq3.15}, it can be seen that
\begin{equation} \label{eq3.18}
\{\hat L_\FKP(q)\}(\bfy) =  \curV(\bfy)\, q(\bfy) - \frac{1}{2} \nabla^2 q(\bfy) \quad , \quad \bfy\in\RR^\nu \, ,
\end{equation}
in which $\nabla^2$ is the Laplacian operator in $\RR^\nu$ and where $\bfy\mapsto\curV(\bfy)$ is the function from $\RR^\nu$ into $\RR$, which is defined, for all $\bfy$ in $\RR^\nu$, as
\begin{equation} \label{eq3.19}
 \curV(\bfy) =  \frac{1}{8} \Vert \,\nabla\Phi(\bfy) \,  \Vert^2  -  \frac{1}{4} \nabla^2 \Phi(\bfy) \, .
\end{equation}
%
\subsection{Properties of operator $\hat L_\FKP$ and hypothesis on its spectrum}
\label{Section3.5}

Let $\delta q$ be a function from $\RR^\nu$ into $\RR$, belonging to the admissible set that allows the evaluation of the bracket
\begin{equation} \nonumber
\langle \hat L_\FKP (q) \, , \delta q\rangle  = \int_{\RR^\nu} \{\hat L_\FKP(q)\}(\bfy) \,\, \delta q(\bfy) \, d\bfy \, .
\end{equation}
Removing $\bfy$ and using Eq.~\eqref{eq3.15} with Eq.~\eqref{eq3.9} yields
\begin{equation} \label{eq3.21}
\langle \hat L_\FKP (q) \, , \delta q\rangle  = -\frac{1}{2} \int_{\RR^\nu}
(p_\bfH^{-1/2} \, \delta q) \, \Div \{ p_\bfH \, \nabla ( p_\bfH^{-1/2}\, q)\}\, d\bfy \, .
\end{equation}
Using the condition at infinity, defined by Eq.~\eqref{eq3.17}, Eq.~\eqref{eq3.21} can be rewritten as,
\begin{equation} \label{eq3.22}
\langle \hat L_\FKP (q) \, , \delta q\rangle  = \frac{1}{2} \int_{\RR^\nu}
p_\bfH\, \langle \nabla(p_\bfH^{-1/2} \, q) \, , \nabla( p_\bfH^{-1/2} \, \delta q) \rangle_{\RR^\nu} \, d\bfy \, .
\end{equation}

\noindent (a) Eq.~\eqref{eq3.22} shows that $\hat L_\FKP$ is a symmetric and positive operator.

\noindent (b) Eqs.~\eqref{eq3.10} and \eqref{eq3.15} show that
\begin{equation} \label{eq3.23}
\hat L_\FKP (q_0) = 0 \quad \hbox{for} \quad q_0 = p_\bfH^{1/2} \, .
\end{equation}
Since $p_\bfH(\bfy)\, d\bfy$ is a bounded positive measure (probability measure), the right-hand side of Eq.~\eqref{eq3.22} shows that the null space  of $\hat L_\FKP$ is of dimension $1$ and is constituted of the function $q_0 = p_\bfH^{1/2}$. For $q=\delta q \not = q_0$, and $\Vert q_0\Vert \not = 0$, we have $\langle \hat L_\FKP (q) \, ,  q\rangle\, > 0$. Therefore, $\hat L_\FKP$ is a positive operator (in the quotient space by the null space).

It is assumed that $p_\bfH$ defined by Eq.~\eqref{eq2.3}, which is constructed with the $n_d$ points $\{\bfeta^j, j=1,\ldots, n_d\}$ of the training dataset, is such that the spectrum of $\hat L_\FKP$ is countable. Due to (a) and (b), we then deduce that the eigenvalues of $\hat L_\FKP$ (defined by Eqs.~\eqref{eq3.16} and \eqref{eq3.17}) are positive except one that is zero. We will also assume that the multiplicity of each eigenvalue is finite and that the eigenfunctions are continuous on $\RR^\nu$.

\subsection{Eigenvalue problem for operator $\hat L_\FKP$ and objective}
\label{Section3.6}

Under the hypothesis introduced in Section~\ref{Section3.5}, the eigenvalue problem
\begin{equation} \label{eq3.23b}
\hat L_\FKP(q_m) = \lambda_m\, q_m \quad , \quad , m\in\NN\, ,
\end{equation}
for operator $\hat L_\FKP$, with the condition defined by Eq.~\eqref{eq3.17}, is such that
\begin{equation} \label{eq3.24}
 0 =\lambda_0 < \lambda_1 \leq \lambda_2 \leq \ldots \, ,
\end{equation}
the multiplicity of each eigenvalue being finite and the associated eigenfunctions $\{q_m ,m\in\NN\}$ belong to $L^2(\RR^\nu)$.
We will admit that the family $\{q_m, m\in\NN\}$ is a Hilbert basis of $L^2(\RR^\nu)$. We then have
\begin{equation} \label{eq3.26}
  \langle q_m , q_n \rangle_{L^2} = \int_{\RR^\nu} q_m(\bfy) \, q_n(\bfy)\, d\bfy = \delta_{mn} \, .
\end{equation}
The eigenfunction $q_0$ associated with $\lambda_0=0$, is such that (see Eq.~\eqref{eq3.23}),
\begin{equation} \label{eq3.27}
  q_0 = p_\bfH^{1/2} \quad ,\quad \Vert\, q_0\Vert_{L^2} = 1 \, ,
\end{equation}
and we have
\begin{equation} \label{eq3.28}
  \sum_{m\in\NN} q_m(\bfy) \, q_m(\bfx)\, d\bfy = \delta_0(\bfy-\bfx) \, .
\end{equation}
From Eqs.~\eqref{eq3.26} and \eqref{eq3.27}, it can be deduced that
\begin{equation} \label{eq3.28}
  \forall m  \geq 1 \quad , \quad  \int_{\RR^\nu} p_\bfH(\bfy)^{1/2} \, q_m(\bfy) \, d\bfy = 0 \, .
\end{equation}
Here, the objective is to compute with quantum computing algorithm, the values of the eigenfunctions $q_m$,
for $m = 1,\ldots , M$, associated with the eigenvalues $0=\lambda_0 < \lambda_1 <\ldots < \lambda_M$ with $M < n_d$ at the sampling points $\bfeta^1, \ldots ,\bfeta^{n_d}$, that is to say, for fixed $m$, $q_m(\bfeta^1), \ldots , q_m(\bfeta^{n_d})$.

\subsection{Proposed methodology to construct a representation of the potential $\curV$ adapted to quantum computing algorithm}
\label{Section3.7}
As we explained in Section~\ref{Section1}, within the formulation we propose for the quantum computing algorithm, it is necessary to construct a polynomial representation in $\bfy$ of the potential $\curV(\bfy)$. Clearly, the expression of $\curV(\bfy)$, as defined by Eq.\eqref{eq3.19}, where $\Phi(\bfy)$ is defined by Eq.\eqref{eq2.8}, is not a polynomial. Therefore, we need to construct an approximation of $\curV(\bfy)$ in the space of multivariate polynomials on $\RR^\nu$, for instance, using monomials on $\RR^\nu$. Such a construction is not straightforward for the following reasons:

\noindent (i) The problem is set in high dimensions, meaning that $\nu$ is large.

\noindent (ii) The sequence of polynomials chosen from the selected subset of all the polynomials on $\RR^\nu$ must be complete in the adapted vector space, and the convergence must be fast with respect to the number of polynomials and their maximum  degrees.

\noindent (iii) Eq.~\eqref{eq3.19} shows that potential $\curV(\bfy)$ is a linear combination of a positive-valued term $\Vert \,\nabla\Phi(\bfy) \,  \Vert^2$ and a real-valued term $\nabla^2 \Phi(\bfy)$, where $\Phi(\bfy)$ is defined by Eq.~\eqref{eq2.8}.
Although $\curV(\bfy)$ is not, {\textit{a priori}}, positive for all $\bfy$ in $\RR^\nu$,
there is an  underlying algebraic structure in the expression of $\curV(\bfy)$,  and these induced algebraic properties must be preserved.\\

There is an important property of the problem that can and will be used to select a well-adapted family of polynomials. The probability density function $p_\bfH$ corresponds to the normalized random variable $\bfH$, which is therefore centered and whose covariance matrix is equal to the identity matrix.  Thus, if the training dataset consisted of $n_d$ independent realizations of a normalized Gaussian random vector, then as $n_d \rightarrow +\infty$, $p_\bfH$ would be Gaussian, and $p_\bfH$ would coincide with the canonical normal density on $\RR^\nu$, denoted $p_\bfY$.
In this case, we would have $\Phi(\bfy) = \Vert\bfy\Vert^2 /2$, leading to
$\Vert \nabla\Phi(\bfy) \Vert^2 = \Vert\bfy\Vert^2$ and $\nabla^2\Phi(\bfy) = \nu$, which shows that $\curV(\bfy)$ would be written as $\Vert\bfy\Vert^2/8 - \nu/4$. Consequently, the choice of normalized Hermite polynomials, which constitute a Hilbert basis in the
Hilbert space of all square integrable functions on $\RR^\nu$ for the measure $p_Y(\bfy) d\bfy$, is certainly a good choice with respect
to  (i) and (ii). Therefore, the choice will  be the polynomial chaos expansion (PCE) in the Gaussian space. Nevertheless, to satisfy (iii), it is preferable to construct such a PCE for $f(\bfy) = \log (p_\bfH(\bfy))$ and then, to deduce of it the PCE of $\Vert \nabla f(\bfy) \Vert^2$ and of $\nabla^2 f(\bfy)$, which allows for the PCE of $\curV(\bfy)$ to be obtained, knowing that $\nabla \Phi(\bfy) = -\nabla f(\bfy)$ and $\nabla^2 \Phi(\bfy)= \nabla^2 f(\bfy)$ and so
\begin{equation} \label{eq3.29}
 \curV(\bfy) =  \frac{1}{8} \Vert \,\nabla f(\bfy) \,  \Vert^2  +  \frac{1}{4} \nabla^2 f(\bfy) \, .
\end{equation}
We will then use the Gaussian Sobolev space \cite{Malliavin1995}. Finally, the following comments should be noted:

\noindent (a) Performing a PCE of $p_H$, instead of $f(\bfy) = \log (p_\bfH(\bfy))$,  would not be convenient because it is difficult to preserve the positivity of the PCE of $p_\bfH$ for all $\bfy$ in $\RR^\nu$ with a small number of polynomials of low degrees.

\noindent (b) Performing a direct construction of the PCE of $\nabla\Phi(\bfy)$, and even more so for $\nabla^2\Phi(\bfy)$, would also be difficult to satisfy (ii) because
$\nabla\Phi(\bfy) = - p_\bfH(\bfy)^{-1} \, \nabla p_\bfH(\bfy)$, showing that $p_\bfH(\bfy)$ appears to the denominator.
%
\subsection{Use of the eigenfunctions of the FKP operator in PLoM}
\label{Section3.8}
The PLoM approach was originally developed using a reduced-order diffusion-maps basis (RODB) $\{\bfg_\DM^1,\ldots ,$ $\bfg_\DM^{m_\optpp}\}$  in $\RR^{n_d}$, represented by the matrix $[g_\DM]\in \MM_{n_d,m_\optpp}$ with $m_\optp < n_d$.
Such an RODB is constructed using the isotropic kernel built from the training dataset $\curD_\training(\bfeta)$, which consists of the $n_d$ points $\bfeta^1,\ldots,\bfeta^{n_d}$ in $\RR^d$ (see \cite{Soize2016,Soize2020c}). The optimal dimension
$m_\optp$ is determined as explained in \cite{Soize2022a} or possibly by using the algorithm presented in \cite{Soize2019b}.

If the data in the training dataset are heterogeneous, we proposed in \cite{Soize2024b} to replace the RODB with a reduced-order basis
$\{\bfg_\TB^1,\ldots , \bfg_\TB^{m_\optpp}\}$ (ROTB) in $\RR^{n_d}$, represented by the matrix $[g_\TB]\in \MM_{n_d,m_\optpp}$. This ROTB is constructed using a transient anisotropic kernel based on  the transition probability density function, which is the transient solution of the  Fokker-Plank equation defined by Eq.~\eqref{eq3.8} and estimated by solving the It\^o stochastic differential equation defined by Eqs.~\eqref{eq3.1} and \eqref{eq3.2}. It should be noted that in this construction, we introduced a scaling change in the transient anisotropic kernel to connect the ROTB to the RODB in the vicinity of the initial time $t=0$. This approach was proposed because the use of eigenfunctions from the eigenvalue problem of the $L_\FKP$ operator, defined by Eq.~\eqref{eq3.12}, could not be explored due to the fact that, in high dimensions, such an eigenvalue problem cannot be solved using classical computers.

This is the reason why, in the present paper, we explore a methodology for calculating the eigenfunctions associated with the  smallest eigenvalues of the $L_\FKP$ operator using quantum computing, with the aim of developing an approach that could allow, in the near future, replacing the ROTB with the reduced-order basis $\{\bfg_\FKP^1,\ldots , \bfg_\FKP^{m_\optpp}\}$ in $\RR^{n_d}$, represented by the matrix $[g_\FKP]\in \MM_{n_d,m_\optpp}$.
For $m = 1,\ldots,m_\optp$, the vector  $\bfg_\FKP^m$ in $\RR^{n_d}$ is constructed as the value of the eigenfunction $\bfy\mapsto \vc_m(\bfy)$ associated with the eigenvalues $\lambda_m$ at the point $\bfeta^j$ in $\RR^d$. From Eqs.~\eqref{eq3.12}, \eqref{eq3.14},  and \eqref{eq3.23b} with \eqref{eq3.24}, we obtain the component $\bfg_{\FKP,j}^m$   of vector $\bfg_{\FKP}^m \in\RR^{n_d}$, where $j\in \{1,\ldots, n_d\}$,
\begin{equation} \label{eq3.30}
\bfg_{\FKP,j}^m = \vc_m(\bfeta^j) = p_\bfH(\bfeta^j)^{1/2}\, q_m(\bfeta^j)\, .
\end{equation}
%

%
\section{Polynomial chaos expansion of potential $\curV$ in the Gaussian space for quantum computing algorithm}
\label{Section4}
%
In this section, we introduce the functional spaces necessary to construct the Gaussian polynomial chaos expansion. We then present and prove such an expansion of $\curV$ based on the strategy proposed in Section~\ref{Section3.7}. Finally we introduce the convergence criteria of the truncated Chaos representation of $\curV$.
%
\subsection{Properties of the function $f = log(p_\bfH)$ and its two first derivatives}
\label{Section4.1}
We begin with Definition~\ref{definition:1} of functional spaces and then proceed with  Proposition~\ref{proposition:1}, which provides the properties of the function $f = log(p_\bfH)$ and its two first derivatives.

\begin{definition}[Hilbert spaces $L^2(\Theta,\RR)$ and $\HH= L^2(\RR^\nu; p_\bfY(\bfy)\, d\bfy)$ ] \label{definition:1}
(i) Let $\bfY=(Y_1,\ldots,Y_\nu)$ be a $\RR^\nu$-valued normalized Gaussian random variable defines on a probability space $(\Theta,\curT,\curP)$, whose probability measure $p_\bfY(\bfy)\, d\bfy$ is defined by the canonical normal density $\bfy\mapsto p_\bfY(\bfy)$ on $\RR^\nu$,
\begin{equation} \label{eq4.1}
p_\bfY(\bfy) =  (2\pi)^{-\nu/2} \exp(-\Vert \bfy \Vert^2 /2) \, .
\end{equation}
Let $L^2(\Theta,\RR)$ be the Hilbert space of all second-order real-valued random variables defined on the probability space $(\Theta,\curT,\curP)$ (equivalence classes of almost-surely equal random variables), equipped with the inner product and the associated norm,
\begin{equation} \label{eq4.2}
\langle U , V \rangle_\Theta \, = E\{U V\} = \int_\Theta U(\theta)\, V(\theta) \, d\curP(\theta)
\quad , \quad
 \Vert \, U\, \Vert_\Theta \,  = (E\{ U^2\} ) ^{1/2} \, .
\end{equation}
(ii) Let $\HH = L^2(\RR^\nu; p_\bfY(\bfy)\, d\bfy)$ be the Hilbert space of all the real-valued functions on $\RR^\nu$, which are square integrable on $\RR^\nu$ with respect to the measure $p_\bfY(\bfy) \,d\bfy$, equipped with the inner product and the associated norm,
\begin{equation} \label{eq4.3}
\langle u , \vc \rangle_\HH  = \int_{\RR^\nu} u(\bfy)\, \vc(\bfy)\, p_\bfY(\bfy)\, d\bfy
\quad , \quad
 \Vert \, u\, \Vert_\HH^2 = \angle u , u \rangle_\HH \, .
\end{equation}
For all continuous functions $u$ and $\vc$  in $\HH$, $U = u(\bfY)$ and $V = \vc(\bfY)$ arereal-valued random variables  in $L^2(\Theta,\RR)$,
\begin{equation} \label{eq4.4}
\langle u , \vc \rangle_\HH  = E\{ u(\bfY) \, \vc(\bfY)\} = E\{U V\} = \langle U , V \rangle_\Theta
\quad , \quad
 \Vert \, u\, \Vert_\HH^2 = E\{ u(\bfY)^2 \} = E\{ U^2\} ) = \Vert \, U\, \Vert_\Theta^2  \, .
\end{equation}
\end{definition}

%
\begin{proposition}[Properties of function $f = \log(p_\bfH)$] \label{proposition:1}
Let $\nu \geq 1$ and $n_d\geq 1$ be fixed integers, and let $\curD_\training(\bfeta)$ be the training dataset of points $\bfeta^1,\ldots,\bfeta^{n_d}$ in $\RR^\nu$, defined in Section~\ref{Section2.2}. Let $p_\bfH$ and $p_\bfY$ be defined by Eqs.~\eqref{eq2.7} and \eqref{eq4.1}.

\noindent (i) The real-valued function $f$ is continuous on $\RR^\nu$ and belongs to $\HH$,
\begin{equation} \label{eq4.5}
E\{f(\bfY)^2\} = \Vert \, f\, \Vert_\HH^2 \,  < + \infty \, .
\end{equation}
(ii) Let $\bfg = (g_1,\ldots, g_\nu)$ be the gradient  function defined by $\bfg = \nabla f = (\partial_1 f,\ldots , \partial_\nu f)$ where $\partial_j = \partial/\partial y_j$ is a $\RR^\nu$-valued continuous function on $\RR^\nu$. For all $j=1,\ldots,\nu$, the function  $g_j$ belongs to $\HH$ and the function $\hat g = \Vert\,\bfg\, ||$ belongs  to $\HH$,
\begin{equation} \label{eq4.6}
E\{ g_j(\bfY)^2\} = \Vert \, g_j\, \Vert_\HH^2 \, < + \infty
\quad , \quad
E\{  \hat g(\bfY)^2\} = \Vert \, \hat g\, \Vert_\HH^2 \, < + \infty
\quad , \quad
E\{  \hat g(\bfY)^4\}  < + \infty
\end{equation}
(ii) The Laplacian function $h = \nabla^2 f = \partial^2_1 f + \ldots + \partial^2_\nu f$  is a real-valued continuous function on $\RR^\nu$ and the function $h$ belongs  to $\HH$,
\begin{equation} \label{eq4.7}
E\{ \vert\, h(\bfY)\, \vert^2\} = \Vert \, h\, \Vert_\HH^2 \, < + \infty \, .
\end{equation}
\end{proposition}
%
\begin{proof} (Proposition~\ref{proposition:1}).
The continuity of functions $f$, $g$, and $h$ is simple to be prove.

\noindent (i) Let $p_\bfH(\bfy)= c_\nu\, \xi(\bfy)$ be defined by Eq.~\eqref{eq2.7} where
$\xi(\bfy) = \frac{1}{n_d}\sum_{j=1}^{n_d}  \exp ( -\frac{1}{2 {\widehat s\,}^2} \Vert \bfy -
\frac{\widehat s}{s}\bfeta^j\Vert^2 )$.
For all $\bfy$ in $\RR^\nu$, we have $0 < \xi(\bfy) \leq 1$ and
$(\log p_\bfH(\bfy))^2 = ( \log c_\nu + \log \xi(\bfy))^2 \, \leq 2(\log c_\nu)^2 + 2 (\log\xi(\bfy))^2$.
Consequently, $E\{f(\bfY)^2\} \, \leq  2(\log c_\nu)^2 + 2 E\{ (\log\xi(\bfY))^2\}$. Since $z\mapsto (\log z)^2$ is a convex function on $]0,1]$, and since $\xi(\bfy) \in ]0,1]$ for all $\bfy$ in $\RR^\nu$, using the Jensen inequality yields,
$(\log\xi(\bfy))^2 \, \leq \frac{1}{n_d}\sum_{j=1}^{n_d} \{\log  ( \exp ( -\frac{1}{2 \,{\widehat s\,}^2} \Vert \bfy -
\frac{\widehat s}{s}\bfeta^j\Vert^2 ) ) \}^2 = \frac{1}{n_d}\sum_{j=1}^{n_d} \frac{1}{4\, {\widehat s\,}^4} \Vert \bfy - \frac{\widehat s}{s}\bfeta^j\Vert^4$.
Therefore, we have $E\{(\log\xi(\bfY))^2\} < +\infty$, which proves Eq.~\eqref{eq4.5}.

\noindent (ii) We have $\bfg(\bfy) = p_\bfH(\bfy)^{-1}\, \nabla p_\bfH(\bfy) = \xi(\bfy)^{-1} \nabla \xi(\bfy)$ and since
$\nabla\xi(\bfy) = -\frac{1}{n_d} \frac{1}{{\widehat s}^2}\sum_{j=1}^{n_d} (\bfy - \frac{\widehat s}{s}\bfeta^j) \,\exp ( -\frac{1}{2 {\widehat s\,}^2} \Vert \bfy - \frac{\widehat s}{s}\bfeta^j\Vert^2 )$, we deduce that
$\Vert\bfg(\bfy)\Vert \, \leq \frac{1}{{\widehat s}^2}\, \Vert \bfy\Vert + \frac{1}{n_d} \frac{1}{{\widehat s}^2 \xi(\bfy)}\sum_{j=1}^{n_d}
 \frac{\widehat s}{s}\, \Vert\bfeta^j\Vert \, \exp ( -\frac{1}{2 {\widehat s\,}^2} \Vert \bfy - \frac{\widehat s}{s}\bfeta^j\Vert^2 )$.
Since $n_d$ is finite and since $E\{\Vert \bfH\Vert^2\} =\frac{1}{n_d} \sum_{j=1}^{n_d} \Vert \bfeta^j\Vert^2 = \nu$ is also finite, we have $\sup_{j=1,...,n_d} \Vert\bfeta^j\Vert = c_\eta < +\infty$. Therefore,
$\Vert\bfg(\bfy)\Vert \, \leq  \frac{1}{{\widehat s}^2}\, \frac{\widehat s}{s}\, c_\eta + \frac{1}{{\widehat s}^2}\, \Vert\bfy\Vert$ and
$\Vert\bfg(\bfy)\Vert^4 \, \leq  8 (\frac{1}{{\widehat s}^2}\, \frac{\widehat s}{s}\, c_\eta)^4 + 8(\frac{1}{{\widehat s}^2}\, \Vert\bfy\Vert)^4$. Since
$E\{\Vert\bfY\Vert\} < + \infty$, $E\{\Vert\bfY\Vert^2\} < + \infty$, and $E\{\Vert\bfY\Vert^4\} < + \infty$,  we have proven the second  and third parts of  Eq.~\eqref{eq4.6}. Since, for all $j$, $E\{g_j(\bfY)^2\} \leq E\{\hat g(\bfY)^2 \} < +\infty$ yields the first part of Eq.~\eqref{eq4.6}.

\noindent (iii) It can be seen that $h(\bfy) = \nabla^2 f(\bfy) = -\Vert \bfg(\bfy)\Vert^2 \, + \, \xi(\bfy)^{-1}\, \nabla^2\xi(\bfy)$, and
$\xi(\bfy)^{-1}\, \nabla^2\xi(\bfy) =  -\frac{\nu}{{\widehat s\,}^2}
+ \frac{1}{n_d} \frac{1}{{\widehat s}^4\,\xi(\bfy)}\sum_{j=1}^{n_d} \Vert\bfy - \frac{\widehat s}{s}\bfeta^j\Vert^2 \,
\exp ( -\frac{1}{2 {\widehat s\,}^2} \Vert \bfy - \frac{\widehat s}{s}\bfeta^j\Vert^2 )$,
which allows us to write
$\vert\xi(\bfy)^{-1}\, \nabla^2\xi(\bfy)\vert \,
\leq  \frac{\nu}{{\widehat s\,}^2} + \frac{1}{{\widehat s\,}^4} \, \Vert \bfy \Vert^2
+ \frac{1}{n_d} \frac{1}{{\widehat s}^4}\, \frac{{\widehat s\,}^2}{{ s}^2}   \frac{1}{\xi(\bfy)}
\sum_{j=1}^{n_d} \Vert\bfeta^j\Vert^2\,
\exp ( -\frac{1}{2 {\widehat s\,}^2} \Vert \bfy - \frac{\widehat s}{s}\bfeta^j\Vert^2 )$.
Since $\sup_j \Vert\bfeta^j\Vert^2 = \tilde c_\eta < +\infty$, we obtain
$\vert\xi(\bfy)^{-1}\, \nabla^2\xi(\bfy)\vert \,
\leq  \frac{\nu}{{\widehat s\,}^2} + \frac{1}{{\widehat s\,}^4} \,\frac{{\widehat s\,}^2}{{ s}^2} \,\tilde c_\eta
    + \frac{1}{{\widehat s\,}^4} \, \Vert \bfy \Vert^2$.
It can then be deduced that
$\vert h(\bfy)\vert  \, \leq \, \Vert\bfg(\bfy)\Vert^2 +  \frac{\nu}{{\widehat s\,}^2}  +
\frac{1}{{\widehat s\,}^4} \,\frac{{\widehat s\,}^2}{{ s}^2} \,\tilde c_\eta  +
\frac{1}{{\widehat s\,}^4} \, \Vert \bfy \Vert^2$
and using (ii) yields
$\vert h(\bfy)\vert \, \leq \, c_0 + c_1\Vert\bfy\Vert + c_2 \Vert\bfy\Vert^2$ with $c_0 > 0$, $c_1 > 0$, and $c_2 > 0$.
We then have proven Eq.~\eqref{eq4.7} because $E\{\Vert\bfY\Vert^r\} < +\infty$ for $r=0,1,\ldots, 4$.

\end{proof}
%
\subsection{Polynomial chaos expansion in Gaussian space of potential $\curV$}
\label{Section4.2}
We begin with Definition~\ref{definition:2} of the normalized Hermite polynomials on $\RR^\nu$ as a Hilbert basis of $\HH$.
Then, Proposition~\ref{proposition:2} yields the polynomial chaos expansions in $\HH$ on the normalized Hermite polynomials for functions $f$, $\bfg$, and $h$. Finally, Proposition~\ref{proposition:3} provides the Polynomial chaos expansion in Gaussian space for the potential $\curV$.

\begin{definition}[Normalized Hermite polynomials on $\RR^\nu$ as a Hilbert basis of $\HH$ ] \label{definition:2}
(i) {\textit{Normalized Hermite polynomials on $\RR$}}. Let $\{H_n(y), n\in\NN\}$ be the Hermite polynomials on $\RR$ related to the canonical normal measure $p_Y(y)\, dy$ on $\RR$ where  $p_Y(y)=(2\pi)^{-1/2} \exp(-y^2/2)$, whose the first three polynomials are $H_0(y) = 1$, $H_1(y) = y$, and $H_2(y) = y^2 -1$. It is well known that $\{\psi_n(y) = H_n(y)/\sqrt{n!}\, , \, n\in\NN\}$ is a Hilbert basis of the Hilbert space of $L^2(\RR; p_Y(y)\, dy)$ (see for instance \cite{Wiener1938,Cameron1947,Kree1986,Malliavin1995}).

\noindent (ii) {\textit{Normalized Hermite polynomials on $\RR^\nu$}}. Let $\bfalpha = (\alpha_1,\ldots,\alpha_\nu)$ be the multi-index in $\NN^\nu$.  Therefore, the family of normalized Hermite polynomials $\{ \psi_\bfalpha, \bfalpha\in\NN^\nu \}$, defined by
\begin{equation} \label{eq4.8}
\psi_\bfalpha(\bfy) = \frac{H_{\alpha_1}(y_1)}{\sqrt{\alpha_1!}} \times \ldots \times \frac{H_{\alpha_\nu}(y_\nu)}{\sqrt{\alpha_\nu!}} \quad , \quad
\forall \bfy = (y_1,\ldots , y_\nu) \in \RR^\nu\, ,
\end{equation}
constitutes a Hilbert basis of $\HH$, which implies that
\begin{equation} \label{eq4.9}
\langle \psi_\bfalpha , \psi_\bfbeta \rangle_\HH = E\{\psi_\bfalpha(\bfY) \, \psi_\bfbeta(\bfY)\} = \delta_{\bfalpha\bfbeta} \quad , \quad \forall \bfalpha \in\NN^\nu \, , \,
\forall \bfbeta \in\NN^\nu \, .
\end{equation}
Let $\vert\bfalpha\vert = \alpha_1 + \ldots + \alpha_\nu$ be the degree of $\psi_\bfalpha$. For $\vert\bfalpha\vert = 0$ we have
$\bfalpha^{(0)} = (0,\ldots, 0)$ and $\psi_{\bfalpha^{(0)}}(\bfy) = 1$. It can then be deduced that, for all $\bfalpha$ such that
$\vert\bfalpha\vert \geq 1$, $E\{\psi_\bfalpha(\bfY)\}=0$.
\end{definition}
%
\begin{proposition}[Polynomial chaos expansions in $\HH$ for functions $f$, $\bfg$, and $h$] \label{proposition:2}
Let $f$, $\bfg =(g_1,\ldots,g_\nu)$, and $h$ be the functions on $\RR^\nu$ defined in Proposition~\ref{proposition:1}.

\noindent (i) Function $f$ in $\HH$ admits the polynomial chaos expansion
$f = \sum_{\bfalpha\in\NN^\nu}  f_\bfalpha \, \psi_\bfalpha$,
which is  convergent in $\HH$.
The coefficients $f_\bfalpha = \langle f , \psi_\bfalpha \rangle_\HH$ are real, $f_{\bfalpha^{(0)}} = E\{f(\bfY)\}$, and
$\Vert f\Vert_\HH^2 = \sum_{\bfalpha\in\NN^\nu}  f_\bfalpha^2 < +\infty$.

\noindent (ii) For all $j=1,\ldots,\nu$, function $g_j =\partial_j f$ in $\HH$ admits the polynomial chaos expansion
$g_j = \sum_{\bfalpha\in\NN^\nu}  g_{j,\,\bfalpha} \, \psi_\bfalpha$,
which is convergent in $\HH$,
where the real coefficients $g_{j,\,\bfalpha}$ are defined by
\begin{equation} \label{eq4.10}
g_{j,\,\bfalpha} = \sqrt{\alpha_j + 1}\, f_{\alpha_1,\ldots,\alpha_j+1,\ldots,\alpha_\nu }
\quad , \quad \forall \bfalpha= (\alpha_1,\ldots,\alpha_j,\ldots,\alpha_\nu) \in \NN^\nu\, .
\end{equation}
We have $g_{j,\,\bfalpha^{(0)}} = E\{g_j(\bfY)\}$ and
$\Vert g_j \Vert_\HH^2 = \sum_{\bfalpha\in\NN^\nu}  g_{j,\,\bfalpha}^2 < +\infty$.

\noindent (iii) Function $h=\sum_{j=1}^\nu \partial_j^2 f$ in $\HH$ admits the following polynomial chaos expansion
$h = \sum_{\bfalpha\in\NN^\nu}  h_\bfalpha \, \psi_\bfalpha$,
which are convergent in $\HH$, where the real coefficients $h_{\bfalpha}$ are defined by
\begin{equation} \label{eq4.11}
h_{\bfalpha} = \sum_{j=1}^\nu \sqrt{(\alpha_j + 1)(\alpha_j + 2)}\, f_{\alpha_1,\ldots,\alpha_j+2,\ldots,\alpha_\nu }
\quad , \quad \forall \bfalpha= (\alpha_1,\ldots,\alpha_j,\ldots,\alpha_\nu) \in \NN^\nu\, .
\end{equation}
We have $h_{\bfalpha^{(0)}} = E\{h(\bfY)\}$ and
$\Vert h\Vert_\HH^2 = \sum_{\bfalpha\in\NN^\nu}  h_{\bfalpha}^2 < +\infty$.

\end{proposition}
%
\begin{proof} (Proposition~\ref{proposition:2}). We use simultaneously Proposition~\ref{proposition:1} and Definition~\ref{definition:1}.

\noindent (i) Since $f$ is in $\HH$ and since $\{\psi_\bfalpha,\bfalpha\in\NN^\nu\}$ is a Hilbert basis of $\HH$, this  yields (i).

\noindent (ii) Since $g_j$ is in $\HH$, for all $j\in \{1,\ldots , \nu\}$, $f$ has a partial derivative with respect to $y_j$ in $\HH$. Consequently,
$g_j = \partial_j f = \partial_j \sum_{\bfalpha\in\NN^\nu}  f_\bfalpha \, \psi_\bfalpha = \sum_{\bfalpha\in\NN^\nu}  f_\bfalpha \, \partial_j\psi_\bfalpha$. It is known that
$\partial_j H_{\alpha_j}(y_j) = \alpha_j H_{\alpha_j -1}(y_j)$. Therefore, for $\alpha_j\geq 1$, we have $\partial_j H_{\alpha_j}(y_j)/\sqrt{\alpha_j!} =\sqrt{\alpha_j} \, H_{\alpha_j-1}(y_j)/\sqrt{(\alpha_j-1)!} $ and for $\alpha_j=0$, we have
$\partial_j H_{\alpha_j}(y_j)/\sqrt{\alpha_j!} = 0$. Function $g_j$ can then be rewritten as
$$g_j =\sum_{\alpha_1\geq 0} \ldots \sum_{\alpha_j\geq 1} \ldots \sum_{\alpha_\nu\geq 0} \sqrt{\alpha_j} \,  f_{\alpha_1,\ldots , \alpha_j,\ldots , \alpha_\nu}
\, \frac{H_{\alpha_1}}{\sqrt{\alpha_1!}}\times \ldots \times \frac{H_{\alpha_j -1}}{\sqrt{(\alpha_j -1)!}} \times \ldots \times \frac{H_{\alpha_\nu}}{\sqrt{\alpha_\nu!}}\, .$$
Performing the change of index $\alpha_j =\hat\alpha_j + 1$ and renaming  $\hat\alpha_j$ to $\alpha_j$ allow (ii) to be proven.

\noindent (iii) Since $h = \nabla^2 f = \partial^2_1 f + \ldots + \partial^2_\nu f$  is in $\HH$, the function $\partial_j^2 f$ exists in $\HH$. The proof of (iii) is then similar to that of (ii) and is left to the reader.
\end{proof}
%
%
\begin{proposition}[Polynomial chaos expansions in Gaussian space of potential $\curV$] \label{proposition:3}
The potential function $\curV$, defined by Eq.~\eqref{eq3.19} or even by Eq.~\eqref{eq3.29}, is in $\HH$ and its polynomial chaos expansion in Gaussian space  is written for all $\bfy$ in $\RR^\nu$ as
\begin{equation} \label{eq4.12}
\curV(\bfy) = \sum_{\bfalpha'\in\NN^\nu} \curV_{\bfalpha'}\, \psi_{\bfalpha'}(\bfy) \quad , \quad
E\{\Vert \curV(\bfY)\Vert^2 \} = \Vert \curV\Vert_\HH^2 = \sum_{\bfalpha'\in\NN^\nu} \curV_{\bfalpha'}^2 < +\infty \, .
\end{equation}
The polynomial chaos expansion of $\curV(\bfy)$ can also be rewritten as
\begin{equation} \label{eq4.13}
\curV(\bfy) = \frac{1}{8} \sum_{\bfalpha\in\NN^\nu} \sum_{\bfbeta\in\NN^\nu} \ggeclair_{\bfalpha\bfbeta} \, \psi_\bfalpha(\bfy)\, \psi_\bfbeta(\bfy) + \frac{1}{4}\sum_{\bfalpha\in\NN^\nu} h_\bfalpha \, \psi_\bfalpha(\bfy) \, ,
\end{equation}
in which $\ggeclair_{\bfalpha\bfbeta} = \sum_{j=1}^\nu g_{j,\,\bfalpha}\, g_{j,\,\bfbeta}$ and where $g_{j,\bfalpha}$ and $h_{\bfalpha}$ are defined by Eqs.~\eqref{eq4.10} and \eqref{eq4.11}.
For any $\bfalpha'$ in $\NN^\nu$, the real-valued coefficient  $\curV_{\bfalpha'}$ can be calculated by
\begin{equation} \label{eq4.14}
\curV_{\bfalpha'} = \frac{1}{8} \sum_{\bfalpha\in\NN^\nu}\sum_{\bfbeta\in\NN^\nu} \ggeclair_{\bfalpha\bfbeta} \, \,
E\{\psi_\bfalpha(\bfY)\, \psi_\bfbeta(\bfY)\, \psi_{\bfalpha'}(\bfY)\} + \frac{1}{4} h_{\bfalpha'}  \, .
\end{equation}
The mean value of the random variable $\curV(\bfY)$ is such that
\begin{equation} \label{eq4.15}
E\{\curV(\bfY) \} = \frac{1}{8}\sum_{\bfalpha\in\NN^\nu} \ggeclair_{\bfalpha\bfalpha} + \frac{1}{4} \,h_{\bfalpha^{0}} \quad , \quad  \bfalpha^{0} =(0,\ldots,0)\in\NN^\nu\,.
\end{equation}
\end{proposition}
%
\begin{proof} (Proposition~\ref{proposition:3}).
Eq.~\eqref{eq3.29} and notations of Proposition~\ref{proposition:1} yield
$\curV(\bfy) = \frac{1}{8} \, \Vert\bfg(\bfy)\Vert^2 + \frac{1}{4} \, h(\bfy)$.
 Let us prove that $\curV$ is in $\HH$, that is to say, $\curV(\bfY)$ is a second-order random variable. Since $\curV$ is continuous on $\RR^\nu$, $\curV(\bfY)$ is a real-valued random variable.
We have  $\curV(\bfy)^2 \leq  \frac{1}{32} \, \Vert\bfg(\bfy)\Vert^4 + \frac{1}{8} \, h(\bfy)^2$, which shows that
$E\{\curV(\bfY)^2\} \leq  \frac{1}{32} \, E\{\Vert\bfg(\bfY)\Vert^4 \}+ \frac{1}{8} \, E\{h(\bfY)^2\}$. Using Eq.~\eqref{eq4.6} yields
$E\{\curV(\bfY)^2\} = \Vert \curV \Vert_\HH^2 < +\infty$, which proves that $\curV(\bfY)$ is of second order and that $\curV$ is in $\HH$. Consequently, we obtain the three parts of Eq.~\eqref{eq4.12}.
Rewriting $\curV(\bfy)$ as $\curV(\bfy) = \frac{1}{8} \, \sum_{j=1}^\nu \bfg_j(\bfy)^2 + \frac{1}{4} \, h(\bfy)$ and using
Proposition~\ref{proposition:2}-(ii) and (iii) yields Eq.~\eqref{eq4.13}.
Combining the first part of Eq.~\eqref{eq4.12} with Eq.~\eqref{eq4.13}, using Definition~\ref{definition:2}-(ii), and projecting onto $\psi_{\bfalpha'}$ yields   Eq.~\eqref{eq4.14}.
Finally, Eq.~\eqref{eq4.15} is directly deduced from Eq.~\eqref{eq4.13} by replacing $\bfy$ by $\bfY$, taking the mathematical expectation, and using Definition~\ref{definition:2}-(ii).

\end{proof}
%
\begin{remark}[Choice of the polynomial chaos expansion for $\curV$]  \label{remark:1}
The representation defined by Eq.\eqref{eq4.12}, which requires the computation of the coefficients defined by Eq.\eqref{eq4.14}, is given for informational purposes, as this approach is not suitable for the following two reasons. First, and most importantly, it results in the loss of the underlying algebraic structure (presence of a positive term) that we discussed in Section~\ref{Section3.7}-(iii), which is essential to preserve when truncating the representation to a finite number of terms. Second, it necessitates additional computation of the coefficients $\curV_{\bfalpha'}$ using Eq.\eqref{eq4.14}, a computation that is not required with the representation defined by Eq.\eqref{eq4.13}. In conclusion, we will use Eq.~\eqref{eq4.13} going forward, as it is more efficient when truncated to a finite number of terms.
\end{remark}
%
%
\subsection{Truncated polynomial chaos expansion of potential $\curV$ and convergence criterion}
\label{Section4.3}
Following the discussion in Remark~\ref{remark:1}, we consider the representation defined by Eq.~\eqref{eq4.13}, which arises from the substitution  in Eq.~\eqref{eq3.29} of the representation $g_j = \sum_{\bfalpha\in\NN^\nu}  g_{j,\,\bfalpha} \, \psi_\bfalpha$ and
$h = \sum_{\bfalpha\in\NN^\nu}  h_\bfalpha \, \psi_\bfalpha$ established in Proposition~\ref{proposition:2}-(ii) and (iii). Here, $g_{j,\,\bfalpha}$ and $h_\bfalpha$, which  are defined by Eqs.~\eqref{eq4.10} and \eqref{eq4.11}, are derived from the coefficients $f_\bfalpha$ of the polynomial chaos expansion $f = \sum_{\bfalpha\in\NN^\nu}  f_\bfalpha \, \psi_\bfalpha$ established in Proposition~\ref{proposition:2}-(i).
Consequently,  we introduce the truncated polynomial chaos expansion of $f$ and we then deduce the  truncated polynomial chaos expansion of $g_j$ and~$h$.
%
\begin{definition}[Truncated polynomial chaos expansions of $f$, $g_j$, and $h$] \label{definition:3}
Let $\mu$ be a finite integer such that $\mu \geq 2$.
Let $\kappa$ be an integer equal to $0$, $1$, or $2$. We define the subset $A_{\kappa,\mu}$ of $\NN^\nu$ such that
\begin{equation} \label{eq4.16}
A_{\kappa,\,\mu}  = \{ \bfalpha\in\NN^\nu \,\, \vert \,\, 0 \, \leq \, \vert\bfalpha\vert\,  \leq \, \,\mu - \kappa \} \, .
\end{equation}
The number of multi-indices in the set $A_{\kappa,\,\mu}$  is
$\vert A_{\kappa,\,\mu}\vert  = \frac{(\nu+\mu-\kappa)!}{\nu! \,(\mu-\kappa)!}$.
The truncated polynomial chaos expansions $g_j^{A_{1,\,\mu}}$ of  $g_j$ for $j=1,\ldots, \nu$, and $h^{A_{2,\,\mu}}$ of $h$, based on the truncated polynomial chaos expansion $f^{A_{0,\,\mu}}$ of $f$ limited to the maximum degree $\,\mu$, along with their norms in $\HH$, are defined as follows,
\begin{align}
f^{A_{0,\,\mu}} &= \sum_{\bfalpha\in A_{0,\,\mu}} f_\bfalpha \, \psi_\bfalpha \quad , \quad \Vert f^{A_{0,\,\mu}} \Vert_\HH^2 = \sum_{\bfalpha\in A_{0,\,\mu}} f_\bfalpha^2  \label{eq4.17} \\
g_j^{A_{1,\,\mu}} &= \sum_{\bfalpha\in A_{1,\,\mu}} g_{j,\,\bfalpha} \, \psi_\bfalpha \quad , \quad \Vert g_j^{A_{1,\,\mu}} \Vert_\HH^2 = \sum_{\bfalpha\in A_{1,\,\mu}} g_{j,\,\bfalpha}^2 \quad , \quad j=1,\ldots , \nu \label{eq4.18} \\
h^{A_{2,\,\mu}} &= \sum_{\bfalpha\in A_{2,\,\mu}} h_\bfalpha \, \psi_\bfalpha \quad , \quad \Vert h^{A_{2,\,\mu}} \Vert_\HH^2 = \sum_{\bfalpha\in A_{2,\,\mu}} h_\bfalpha^2  \label{eq4.19}
\end{align}
The truncated polynomial chaos expansion $\curV^{\,\mu}$ of $\curV$, based on the truncated polynomial chaos expansion $f^{A_{0,\,\mu}}$ of $f$ limited to the maximum degree $\,\mu$ is defined by,
\begin{equation} \label{eq4.20}
\curV^{\,\mu} (\bfy) = \frac{1}{8} \sum_{\bfalpha\,\in \,A_{1,\,\mu}} \,\,\sum_{\bfbeta\,\in \,A_{1,\,\mu}} \ggeclair_{\bfalpha\bfbeta} \, \psi_\bfalpha(\bfy)\, \psi_\bfbeta(\bfy) + \frac{1}{4}\sum_{\bfalpha\in \,A_{2,\,\mu}} h_\bfalpha \, \psi_\bfalpha(\bfy) \, ,
\end{equation}
\end{definition}
%
\begin{proposition}[Convergence criterion of the sequences in $\mu$] \label{proposition:4}
(i) The sequences of approximations $\{f^{A_{0,\,\mu}}\}_\mu$, $\{g_j^{A_{1,\,\mu}}\}_\mu$ for $j=1,\ldots,\nu$, and $\{h^{A_{2,\,\mu}}\}_\mu$ strongly converge to $f$, $g_j$ for $j=1,\ldots,\nu$, and $h$, respectively, as $\mu$ approaches $+\infty$, and we have
\begin{equation} \label{eq4.21}
\lim_{\mu\rightarrow +\infty} \Vert f^{A_{0,\,\mu}} \Vert_\HH \, = \Vert f\Vert_\HH \quad , \quad
\lim_{\mu\rightarrow +\infty} \Vert g_j^{A_{1,\,\mu}} \Vert_\HH \, = \Vert g_j\Vert_\HH \quad , \quad
\lim_{\mu\rightarrow +\infty} \Vert h^{A_{2,\,\mu}} \Vert_\HH \, = \Vert h\Vert_\HH \, .
\end{equation}
(ii) The sequence of approximation $\{\curV^{\,\mu}\}_\mu$ converges to $\curV$ as  $\mu$ approaches $+\infty$, and we have
\begin{equation} \label{eq4.22}
\lim_{\mu\rightarrow +\infty} \Vert\curV^{\,\mu} \Vert_\HH \, = \Vert \curV\Vert_\HH \, .
\end{equation}
\end{proposition}
%
\begin{proof} (Proposition~\ref{proposition:4}).
(i) The strong convergence in $\HH$, that is, convergence with respect to the norm in $\HH$, is simply a restatement of  Proposition~\ref{proposition:2}. It is well known that if a sequence $x_\mu$ in a Hilbert space $\HH$ strongly converges to
and element $x$ of $\HH$, then the norm $\Vert x_\mu\Vert_\HH$ converges to $\Vert x\Vert_\HH$.

\noindent (ii) Eq.~\eqref{eq4.12} also shows that the sum  $\sum_{\bfalpha'} \curV_{\bfalpha'} \psi_{\bfalpha'}$ strongly converges to $\curV$ in $\HH$. However, we previously explained, we will not use the representation $\curV = \sum_{\bfalpha'\in\NN^\nu} \curV_{\bfalpha'}\, \psi_{\bfalpha'}$ for $\curV$, but instead of the representation defined by Eq.~\eqref{eq4.13}. Therefore, we need to  prove that if Eq.~\eqref{eq4.21} holds, then Eq.~\eqref{eq4.22} follows. In the proof of Proposition~\ref{proposition:3}, we showed that
$\curV(\bfy)^2 \leq  \frac{1}{32} \, \Vert\bfg(\bfy)\Vert^4 + \frac{1}{8} \, h(\bfy)^2$, which allows us to write
$E\{\curV(\bfY)^2\} \leq  \frac{1}{32} \, E\{\Vert\bfg(\bfY)\Vert^4\} + \frac{1}{8} \, E\{h(\bfY)^2\}$. This expression shows that, to complete the proof of Eq.~\eqref{eq4.22}, we need only prove that there exits a  finite positive constant $c$ such that
$E\{\Vert\bfg(\bfY)\Vert^4\} \leq c\, (E\{\Vert\bfg(\bfY)\Vert^2\} )^2$. In the proof of Proposition~\ref{proposition:1}, we obtained
the inequality $\Vert \bfg(\bfy) \Vert \, \leq a + b\, \Vert\bfy\Vert$, where $a$ and $b$ are finite positive constants. We then have
$(E\{\Vert\bfg(\bfY)\Vert^2\})^2  \leq 4 \, a^4 + 4 \, b^4 (E\{\Vert\bfY\Vert^2\})^2$ and
$E\{\Vert\bfg(\bfY)\Vert^4\}  \leq 8\, a^4 + 8\, b^4 E\{\Vert\bfY\Vert^4\}$.
Since $E\{\Vert\bfY\Vert^2\} = m_2 < +\infty$ and $E\{\Vert\bfY\Vert^4\} = m_4 < +\infty$, there exists a finite positive constant
$c$ such that $E\{\Vert\bfg(\bfY)\Vert^4\} \leq c\, (E\{\Vert\bfg(\bfY)\Vert^2\} )^2$, which concludes the proof.
\end{proof}
%
\begin{remark}[Practical criterion for checking the convergence]  \label{remark:2}
The convergence criterion for each considered sequence cannot directly be based on the strong convergence definition or on the  use of the  Cauchy sequence in Hilbert space $\HH$, which would be numerically expensive. This would require the numerical evaluation of a large number of mathematical expectations for a random function on $\RR^\nu$. Therefore, we propose using a criterion based on Eq.~\eqref{eq4.21} with Eqs.~\eqref{eq4.17} to \eqref{eq4.19},  without evaluating the norm of the limit. The criteria we propose involve the study of the functions,
\begin{equation} \label{eq4.23}
\mu \mapsto \Vert f^{A_{0,\,\mu}}\Vert_\HH  =  (\! \sum\limits_{\bfalpha\in A_{0,\,\mu}} \!\! f_\bfalpha^2  )^{1/2} \,\,\, , \, \,\,
\mu \mapsto \Vert g_j^{A_{1,\,\mu}}\Vert_\HH =  (\!\sum\limits_{\bfalpha\in A_{1,\,\mu}} \!\!  g_{j,\,\bfalpha}^2  )^{1/2} \,\, , \,   j=1,\ldots ,\nu \,\,\, , \, \,\,
\mu\mapsto \Vert h^{A_{2,\,\mu}}\Vert_\HH =  (\! \sum\limits_{\bfalpha\in A_{2,\,\mu}}  \!\! h_\bfalpha^2  )^{1/2} \, .
\end{equation}
\end{remark}
%
\subsection{Algebraic formulas for computation and convergence analyses}
\label{Section4.4}
In this Section, we propose a numerical method for the effective calculation of the coefficients $f_\bfalpha$ of the polynomial chaos expansion. We summarize the calculation process for $g_{j,\bfalpha}$, $h_\bfalpha$, and $\ggeclair_{\bfalpha\bfbeta}$, and address the questions related to  convergence control.\\

\noindent {\textit{(i) Computation of an approximation of $f_\bfalpha$ and convergence criterion for the approximation}.
Let $A^{0,\,\mu}$ be the set of multi-indices $\bfalpha=(\alpha_1,\ldots,\alpha_\nu)\in\NN^\nu$, where the maximum degree $\mu$ is fixed. For a given $\bfalpha$ in $A^{0,\,\mu}$, Proposition~\ref{proposition:2}-(i) defines the coefficient $f_\bfalpha$, which we rewrite in terms of a mathematical expectation as $f_\bfalpha = \langle f, \psi_\bfalpha\rangle_\HH = E\{f(\bfY) \, \psi_\bfalpha(\bfY)\}$. This mathematical expectation is classically estimated using $N$ independent realizations $\{\bfy^\ell \in\RR^\nu, \ell =1,\ldots , N\}$ of the $\RR^\nu$-valued normalized Gaussian random variable $\bfY$ (see Definition~\ref{definition:1}), which yields the following approximation $f_\bfalpha^N$ of $f_\bfalpha$,
\begin{equation} \label{eq4.24}
f_\bfalpha^N = \frac{1}{N} \sum_{\ell=1}^N f(\bfy^\ell) \, \psi_\bfalpha(\bfy^\ell) \quad , \quad \bfalpha\in A^{0,\,\mu}\, ,
\end{equation}
in which $f(\bfy^\ell) = \log(p_\bfH(\bfy^\ell))$, computed using Eq.~\eqref{eq2.3}.
\noindent Let $\bfY^1,\ldots , \bfY^N$ be $N$ independent copies of the random vector $\bfY$. Let $F_\bfalpha^N$ be the real-valued random variable defined on probability space $(\Theta,\curT,\curP)$ such that $F_\bfalpha^N = \frac{1}{N} \sum_{\ell=1}^N f(\bfY^\ell) \, \psi_\bfalpha(\bfY^\ell)$. Then (see, for instance, \cite{Serfling1980}), the sequence of real-valued random variables $\{F_\bfalpha^N\}_N$ on $(\Theta,\curT,\curP)$ is convergent in probability to $f_\bfalpha$,
\begin{equation} \label{eq4.25}
\forall \varepsilon > 0 \quad , \quad \lim_{N\rightarrow +\infty} \curP\{\,\vert F_\bfalpha^N - f_\bfalpha \vert \,\,  \geq \, \varepsilon \} = 0 \, .
\end{equation}
There is also the almost sure convergence, thanks to the strong law of large numbers
\begin{equation} \label{eq4.26}
\curP\{ \lim_{N\rightarrow +\infty}  F_\bfalpha^N = f_\bfalpha \} = 1 \, .
\end{equation}
It should be noted that the speed of convergence is proportional to $1/\sqrt{N^2} = 1/N$ and is independent of dimension $\nu$. The quantification of the approximation error could traditionally be estimated using the central limit theorem, which involves the variance of the estimator (see, for instance, \cite{Serfling1980,Givens2013,Soize2017b}).

\noindent Concerning the convergence with respect to $\mu$, it is understood that the value of $N$, determined for a convergence tolerance based on the criterion defined by Eq.\eqref{eq4.25}, is valid for all multi-indices $\bfalpha$ in $A^{0,,\mu}$, and therefore $N$ does not depend on $\bfalpha$. However, since $A^{0,,\mu}$ depends on $\mu$, $N$ will, \textit{a priori}, depend on $\mu$. Under these conditions, the convergence analysis with respect to $\mu$ is carried out using the first equation in Eq.\eqref{eq4.23}, for which the function $\mu \mapsto \Vert f^{A^{0,,\mu}}\Vert_\HH$ will depend on the value of $N$ that has been determined. This means that the value of $N$ is updated for each value of $\mu$ in the convergence analysis with respect to $\mu$.\\

\noindent {\textit{(ii) Computation of the corresponding approximations of $g_{j,\bfalpha}$, $\ggeclair_{\bfalpha\bfbeta}$, $h_\bfalpha$, and $\curV$}. It is assumed that $\mu$ and the associated value of $N$ are fixed to the values identified in (i) above, for which all the approximations $f_\bfalpha^N$ of the coefficient $f_\bfalpha$ are computed for all $\bfalpha$ in $A^{0,\,\mu}$. The corresponding approximations of the coefficients $g_{j,\,\bfalpha}$ for $j=1,\ldots, \nu$, $\ggeclair_{\bfalpha\bfbeta}$, and $h_{\bfalpha}$ defined by Eq.~\eqref{eq4.10}, in Proposition~\ref{proposition:3}, and by \eqref{eq4.11}, respectively, are thus computed by the formulas:
\begin{equation} \label{eq4.27}
g_{j,\,\bfalpha}^N = \sqrt{\alpha_j + 1}\, f_{\alpha_1,\ldots,\alpha_j+1,\ldots,\alpha_\nu}^N
\quad , \quad \forall \bfalpha= (\alpha_1,\ldots,\alpha_j,\ldots,\alpha_\nu) \in A^{1,\,\mu} \quad , \quad j=1,\ldots , \nu\, ,
\end{equation}
where $A^{1,\,\mu}$ is defined by Eq.~\eqref{eq4.16} for $\kappa=1$,
\begin{equation} \label{eq4.28}
\ggeclair_{\bfalpha\bfbeta}^N = \sum_{j=1}^\nu g_{j,\,\bfalpha}^N \, g_{j,\,\bfbeta}^N
\quad , \quad \forall \bfalpha \in A^{1,\,\mu}\quad , \quad \forall \bfbeta \in A^{1,\,\mu} \, ,
\end{equation}
and
\begin{equation} \label{eq4.29}
h_{\bfalpha}^N = \sum_{j=1}^\nu \sqrt{(\alpha_j + 1)(\alpha_j + 2)}\, f_{\alpha_1,\ldots,\alpha_j+2,\ldots,\alpha_\nu}^N
\quad , \quad \forall \bfalpha= (\alpha_1,\ldots,\alpha_j,\ldots,\alpha_\nu) \in A^{2,\,\mu}\, .
\end{equation}
where $A^{2,\,\mu}$ is defined by Eq.~\eqref{eq4.16} for $\kappa=2$.
The corresponding approximation of $\curV^{\,\mu}$ defined by Eq.~\eqref{eq4.20} are thus computed by
\begin{equation} \label{eq4.30}
\curV^{\,\mu,N}(\bfy) = \frac{1}{8} \sum_{\bfalpha\,\in \,A_{1,\,\mu}} \,\,\sum_{\bfbeta\,\in \,A_{1,\,\mu}} \ggeclair_{\bfalpha\bfbeta}^N \, \psi_\bfalpha(\bfy)\, \psi_\bfbeta(\bfy) + \frac{1}{4}\sum_{\bfalpha\in \,A_{2,\,\mu}} h_\bfalpha^N \, \psi_\bfalpha(\bfy) \quad , \quad \forall \, \bfy \in \RR^\nu\, .
\end{equation}
%
\subsection{Adaptation of the truncated polynomial chaos expansion of potential $\curV$ to a quantum computing formulation}
\label{Section4.5}
According to the explanations given in Sections~\ref{Section1} and \ref{Section3.7}, we need to rewrite the right-hand side of Eq.~\eqref{eq4.29} as a sum of monomials, which is straightforward.

Let $\psi_\bfalpha$ be the normalized Hermite polynomial of multi-index $\bfalpha$ defined by Eq.~\eqref{eq4.8}, where $\bfalpha$ belongs to $A_{1,\mu}$ or $A_{2,\mu}$ as defined by Eq.~\eqref{eq4.16}. For $j$ fixed in $\{1,\ldots, \nu\}$, let $\alpha_j$ be the component $j$ of $\bfalpha = (\alpha_1,\ldots , \alpha_\nu)$, and the Hermite polynomial $H_{\alpha_j}(y_j)$ for the component $y_j$ of $\bfy=(y_1,\ldots ,y_\nu)$ is written as
$H_{\alpha_j}(y_j) = \sum_{m_j=0}^{\alpha_j} a_{m_j,\alpha_j} \, y^{m_j}$.
The coefficients $a_{m_j,\alpha_j}$ are known. For instance, if $\alpha_j=2$, then $H_{\alpha_j}(y_j) = y_j^2 -1$ , and consequently,
$a_{0,2} = -1$, $a_{1,2}= 0$, and $a_{2,2}=1$. However, the maximum value of $\bfalpha_j$ is $\mu-1$ if $\bfalpha \in A_{1,\mu}$ and $\mu-2$ if $\bfalpha \in A_{2,\mu}$. To avoid dependence on the upper bound $\alpha_j$ in the sum that defines $H_{\alpha_j}(y_j)$, which allows us to simplify the algebraic representation, we complete, if necessary, by adding zero coefficients. For $\kappa=1$ or $2$, we therefore have,
\begin{equation} \label{eq4.31}
\forall\bfalpha \in A_{\kappa,\mu}\quad , \quad H_{\alpha_j}^{(\kappa)}(y_j) = \sum_{m_j=0}^{\mu-\kappa} a^{(\kappa)}_{m_j,\alpha_j} \, y^{m_j}_j \quad , \quad
\end{equation}
with $a^{(\kappa)}_{m_j,\alpha_j} = a_{m_j,\alpha_j}$ if $m_j\leq \alpha_j$ and, if $\alpha_j+1 > \mu-\kappa$, then
$a^{(\kappa)}_{m_j,\alpha_j} = 0$ for all $m_j\in \{\alpha_j+1,\ldots ,\mu-\kappa\}$.
Let $\bfm=(m_1,\ldots,m_\nu)$ and $\bfm'=(m_1',\ldots,m_\nu')$, the vectors of integers belonging to the set  $\curM_\kappa = \{0,1,\ldots,\mu-\kappa\}^\nu $. Using Eqs.~\eqref{eq4.8} and \eqref{eq4.31}, Eq.~\eqref{eq4.30} can be rewritten as
\begin{equation} \label{eq4.32}
\curV^{\,\mu,N}(\bfy) = \frac{1}{8} \sum_{\bfm\,\in \,\curM_{1}} \,\,\sum_{\bfm'\,\in \,\curM_{1}} \hat\ggeclair_{\bfm\bfm'}^N \, \bfy^\bfm\, \bfy^{\bfm'} + \frac{1}{4}\sum_{\bfm\in \,\curM_{2}} \hat h_\bfm^N \, \bfy^\bfm \quad , \quad \forall \, \bfy \in \RR^\nu\, .
\end{equation}
In Eq.~\eqref{eq4.31}, $\bfy^\bfm = y_1^{m_1}\times \ldots \times y_\nu^{m_\nu}$, and for all $\bfm$ and $\bfm'$ in $\curM_1$, we have
\begin{equation} \label{eq4.33}
 \hat\ggeclair_{\bfm \bfm'}^N = \sum_{\bfalpha\, \in \, A_{1,\mu}}\,\,  \sum_{\bfbeta\, \in\,  A_{1,\mu}} \ggeclair_{\bfalpha\bfbeta}^N
 \, \frac{ a_{\bfm,\bfalpha}^{(1)} } { \sqrt{\bfalpha !} } \, \frac{ a_{\bfm',\bfalpha}^{(1)} } {\sqrt{\bfbeta !}}   \, ,
\end{equation}
where $\bfalpha ! = \alpha_1 ! \times \ldots \times \alpha_\nu !$ and
$a_{\bfm,\bfalpha}^{(1)}  = a_{m_1,\alpha_1}^{(1)} \times \ldots \times a_{m_\nu,\alpha_\nu}^{(\nu)}$.
For all $\bfm$ in $\curM_2$, we have
\begin{equation} \label{eq4.34}
\hat h_\bfm^N = \sum_{\bfalpha\, \in \, A_{2,\mu}}\, h_\bfalpha^N \, \frac{a_{\bfm,\bfalpha}^{(2)} }{\sqrt{\bfalpha !}}\, .
\end{equation}
%
%
%
\section{Numerical validation of the software code for polynomial-chaos-expansion coefficient computation}
\label{Section5}
%
A software code has been developed based on  the algebraic formulas for computation and convergence analyses, which we have presented and summarized in Section~\ref{Section4.4}.
However, such a validation is not easy if the probability measure $p_\bfH(\bfy), d\bfy$ is arbitrary, because in that case, we do not have a reference. Since we are here to validate a software code, we can choose any probability measure, particularly one for which we have a reference. Therefore, we will take $\bfH$ to be a $\RR^\nu$-valued normalized Gaussian random variable. The choice of a normalized random variable is justified by the fact that Eq.~\eqref{eq2.2} must be satisfied. We will begin this section by providing the reference algebraic formulas for this Gaussian case, and then we will continue by presenting the numerical results. However, we must first specify the parameterization of the numerical analysis that we are going to present.

(i) Since $\bfH$ is chosen as Gaussian for this validation, we can skip the convergence analysis with respect to $\mu$ for the following reason. For the considered validation case, the reference probability density function is $p_\bfH^\exactp(\bfy) =  (2\pi)^{-\nu/2} \exp(-\Vert \bfy \Vert^2 /2)$ that is independent of $n_d$, while its GKDE approximation is given by Eq.~\eqref{eq2.3} and depends on $n_d$. This case defined by $p_\bfH^\exactp$ will be named the Gaussian reference case.
We then have
\begin{equation}  \label{eq5.1}
f (\bfy) = \log(p_\bfH^\exactp(\bfy)) = -\frac{\nu}{2} \log(2\pi) - \frac{1}{2} \,\Vert \bfy \Vert^2 \, ,
\end{equation}
which shows that $\mu = 2$ could be sufficient. Under this condition, we have fixed the value $\mu = 4$.

(ii) The value of $\mu$ being fixed, in order to minimize the number of figures presented for the convergence analyses with respect to $N$  we will limit the calculated quantities to the functions defined by Eq.~\eqref{eq4.23}. Since $\mu$ is fixed and since these functions depend on $N$, we adapt the definition and notation of these functions as follows,
\begin{align}
N \mapsto \curE_f (N) \,  &=  \{\Sigma_{\bfalpha\,\in\, A_{0,\,\mu}}  (f_\bfalpha^N)^2  \}^{1/2}                      \, ,\label{eq5.2}\\
N \mapsto \curE_\bfg (N)\, &= \{\Sigma_{\bfalpha\,\in\, A_{1,\,\mu}} \Sigma_{j=1}^\nu (g_{j,\,\bfalpha}^N)^2\}^{1/2} \, , \label{eq5.3} \\
N \mapsto \curE_h (N) \, &=  \{\Sigma_{\bfalpha\,\in\, A_{2,\,\mu}}  (h_\bfalpha^N)^2  \}^{1/2}                       \, . \label{eq5.4}
\end{align}
This choice of functions defined by Eqs.~\eqref{eq5.2} to \eqref{eq5.4} allows us to control the convergence of all the coefficients of the polynomial chaos expansions. For this Gaussian reference case, we will also use the GKDE representation. In that case, these functions also depend on $n_d$, and we will consider these functions as a family of functions indexed by the parameter $n_d$.
%
\subsection{Algebraic formulas for the Gaussian reference case}
\label{Section5.1}
Using Eq.~\eqref{eq4.30}, it can be seen that $\bfg(\bfy) = \nabla f(\bfy) = -\bfy$ and
$h(\bfy) = \nabla^2 f(\bfy) = -\nu$. Using the classical formulas: $E\{Y_j\} = 0$, $E\{Y_j^2\} = 1$, and $E\{Y_j^4\} = 3$, it can easily proven the following results,
\begin{align}
\curE^\exactp_f & = \Vert\,  f \, \Vert_\HH = \{\, \Sigma_{\bfalpha\,\in\, \NN^\nu}\,\,  f_\bfalpha^2 \, \}^{1/2} =
  \{ {\nu}/2 +  ( 1\! + \! 2 \log(2\pi)\! +\! (\log(2\pi))^2 \,)\, \nu^2/4\,   \}^{1/2}                              \, , \label{eq5.5} \\
\curE^\exactp_\bfg & = \{\, \Sigma_{j=1}^\nu \, \Vert g_j \Vert_\HH^2 \,\}^{1/2}
               = \{\, \Sigma_{j=1}^\nu \Sigma_{\bfalpha\,\in\, \NN^\nu} \,\, g_{j,\bfalpha}^2 \, \}^{1/2} = \sqrt{\nu}  \, , \label{eq5.6} \\
\curE^\exactp_h & = \Vert \, h \, \Vert_\HH = \{ \,\Sigma_{\bfalpha\,\in\, \NN^\nu} \,\, h_\bfalpha^2 \, \}^{1/2} = \nu \, . \label{eq5.7}
\end{align}

\subsection{Numerical validation of the software code}
\label{Section5.1}
The convergence analysis with respect to $N$ is done for two  values of $\nu$. Fig.~\ref{fig:figure1} concerns the results for $\nu=1$, while Fig.~\ref{fig:figure2} concerns the results for $\nu=10$.\\

\noindent {\textit{ (i) Gaussian reference case}. For the Gaussian reference case, for which the probability density function of $\bfH$ is $p_\bfH^\exactp$ defined at the beginning of Section~\ref{Section5}, the use of formulas defined by Eqs.~\eqref{eq5.5} to \eqref{eq5.7} yields $(\curE^\exactp_f,\curE^\exactp_\bfg,\curE^\exactp_h) = (1.585, 1.000, 1.000)$ for $\nu=1$ and $(14.364,3.162,10.000)$ for $\nu=10$.
For analyzing the convergence of the  polynomial chaos expansions of this Gaussian reference case, the functions $N\mapsto \curE_f (N), \curE_\bfg (N)$, and $\curE_f (N)$ are computed using Eqs.~\eqref{eq5.2}, \eqref{eq5.3}, and \eqref{eq5.4}, respectively. The graphs of these functions are displayed in Figs.~\ref{fig:figure1a}, \ref{fig:figure1b}, and \ref{fig:figure1c} for $\nu=1$, and in
Figs.~\ref{fig:figure1d}, \ref{fig:figure1e}, and \ref{fig:figure1f} for $\nu=10$.
For $N={10}^7$, we have $(\curE_f ({10}^7), \curE_\bfg ({10}^7), \curE_f ({10}^7)) = (1.585, 0.999, 0.999)$ for $\nu = 1$ and
$(14.364,3.171,9.971)$ for $\nu=10$. These values are very close to the reference values that show a complete convergence and a full validation of the software code devoted to the computation of the polynomial-chaos-expansion coefficients based on the theory presented in Section~\ref{Section4}.\\

\noindent {\textit{ (ii) Gaussian case with GKDE representation}.
With respect to software code validation, this case does not provide anything more than the Gaussian reference case presented in (i), because the algorithms for calculating the coefficients of polynomial chaos expansions are independent of the choice of the probability density of $\bfH$. What this analysis shows for this Gaussian case is the influence of the number $n_d$ of points in the training dataset on the estimation of $p_\bfH$ by the GKDE method (given by Eq.~\eqref{eq2.3}), compared to the exact expression $p_\bfH^\exactp$ (which corresponds to $p_\bfH$ as $n_d\rightarrow +\infty$).
Therefore, for the Gaussian case with the GKDE estimation of $p_\bfH$ defined by Eq.~\eqref{eq2.3}, we have considered four values of $n_d$, which are $100$, $1000$, $5000$, and $10\, 000$.
For the convergence analysis of the polynomial chaos expansions, the functions $N\mapsto \curE_f (N), \curE_\bfg (N)$, and $\curE_f (N)$ are computed using Eqs.~\eqref{eq5.2}, \eqref{eq5.3}, and \eqref{eq5.4}, respectively. The graphs of these functions are displayed in Figs.~\ref{fig:figure2a}, \ref{fig:figure2b}, and \ref{fig:figure2c} for $\nu=1$, and in
Figs.~\ref{fig:figure2d}, \ref{fig:figure2e}, and \ref{fig:figure2f} for $\nu=10$.
For $N={10}^7$, we have $(\curE_f ({10}^7), \curE_\bfg ({10}^7), \curE_f ({10}^7)) = (1.594, 1.026, 1.027)$ for $\nu = 1$ and
$(15.249,4.502,14.148)$ for $\nu=10$.
These values are close to the reference values for $\nu=1$ and slightly further away for $\nu=10$. This does not call into question either the theoretical formulation proposed or the validation of the software code. This difference is simply due to the convergence of $p_\bfH$ towards $p_\bfH^\exactp$ as a function of $n_d$ for a fixed value of $\nu$. We know that the larger $\nu$ is, the larger $n_d$ must be. However, in the context of the probabilistic learning on manifolds approach, the eigenfunctions associated with the first eigenvalues of the Fokker-Planck operator that we are looking for are only used to construct a vector basis for the projection of the data, and replacing the exact operator with an approximate operator, constructed with a small training dataset, does not introduce any difficulties (see, for example, \cite{Soize2024b}).
\begin{figure}[h]
    \centering
    \begin{subfigure}[b]{0.33\textwidth}
    \centering
        \includegraphics[width=\textwidth]{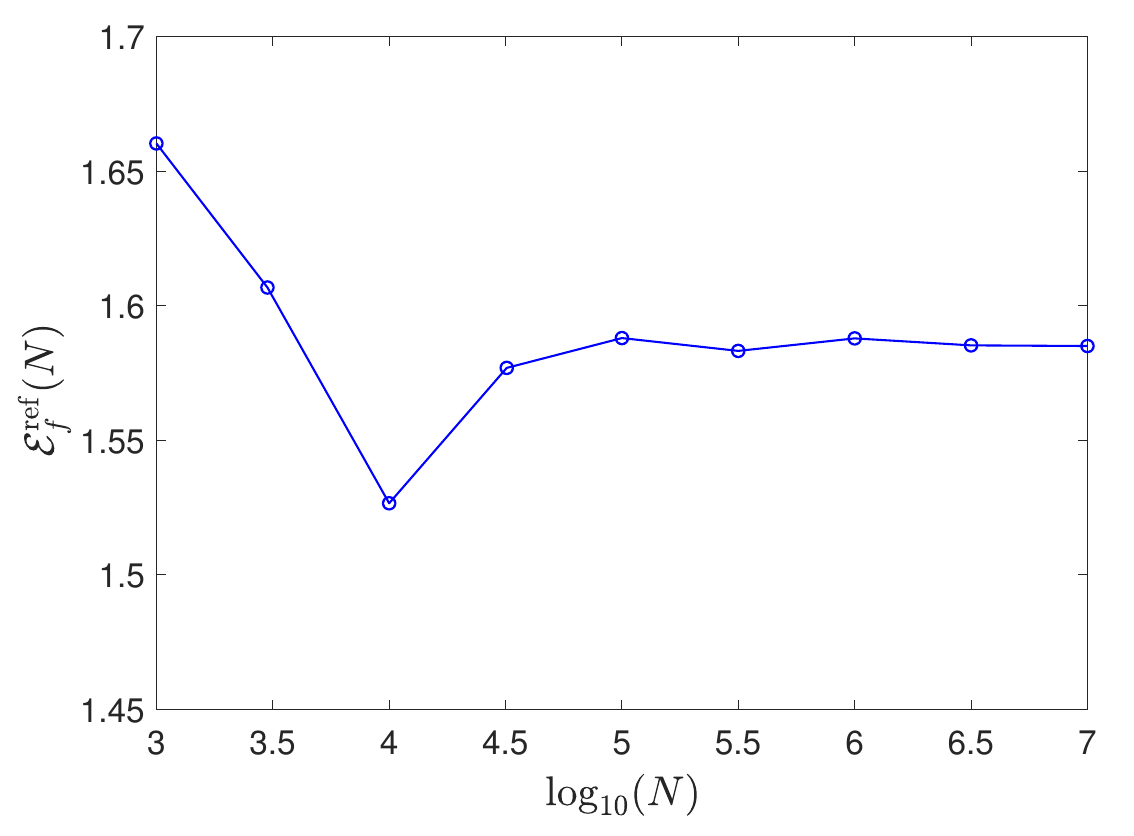}
        \caption{$N \mapsto \curE_f^\exactp (N)$ for $\nu=1$.}
         \vspace{0.3truecm}
        \label{fig:figure1a}
    \end{subfigure}
    \begin{subfigure}[b]{0.33\textwidth}
        \centering
        \includegraphics[width=\textwidth]{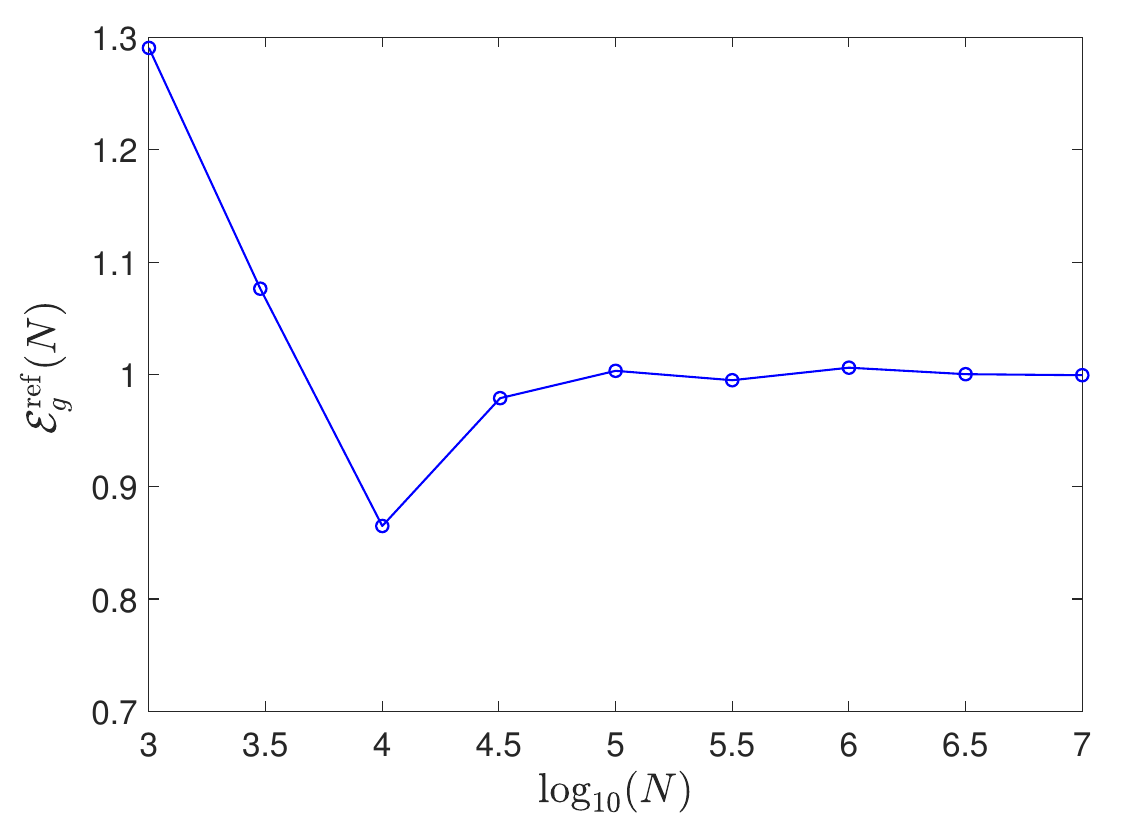}
         \caption{$N \mapsto \curE_\bfg^\exactp (N)$ for $\nu=1$.}
          \vspace{0.3truecm}
        \label{fig:figure1b}
    \end{subfigure}
    \begin{subfigure}[b]{0.33\textwidth}
        \centering
        \includegraphics[width=\textwidth]{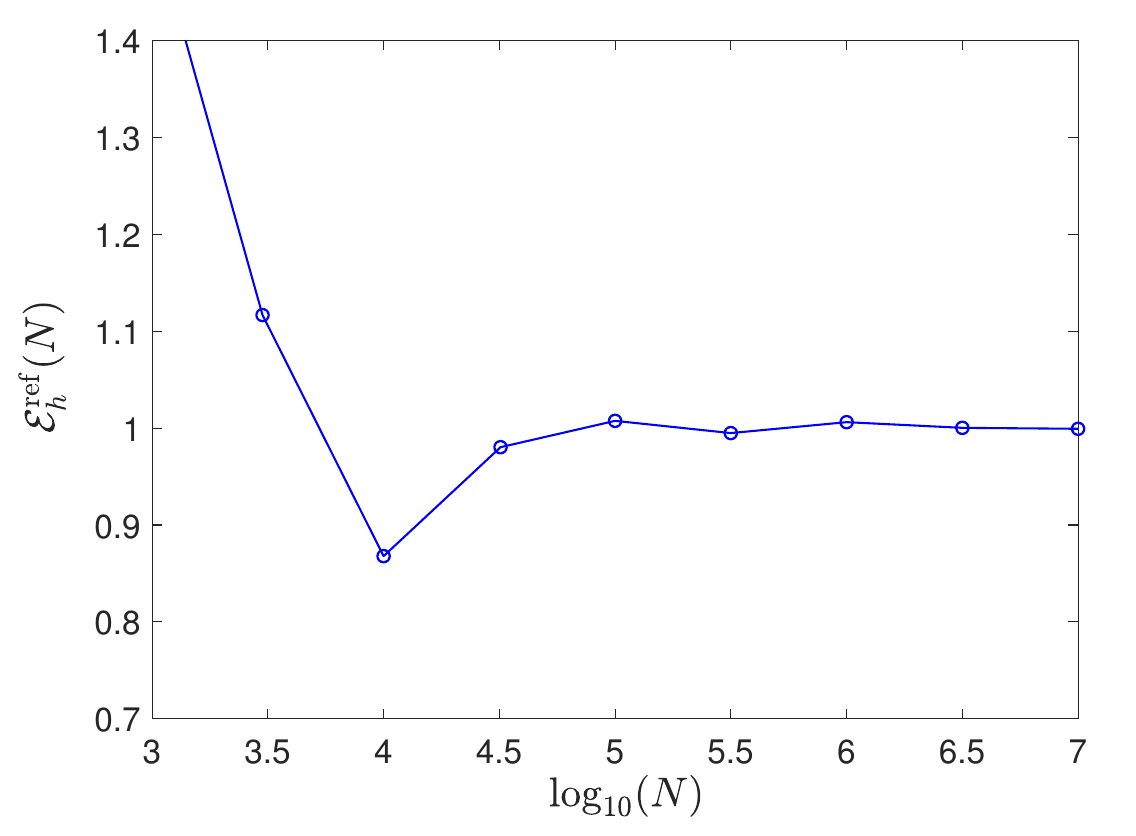}
         \caption{$N \mapsto \curE_h^\exactp (N)$ for $\nu=1$.}
        \label{fig:figure1c}
    \end{subfigure}
    %
    \centering
    \begin{subfigure}[b]{0.33\textwidth}
    \centering
        \includegraphics[width=\textwidth]{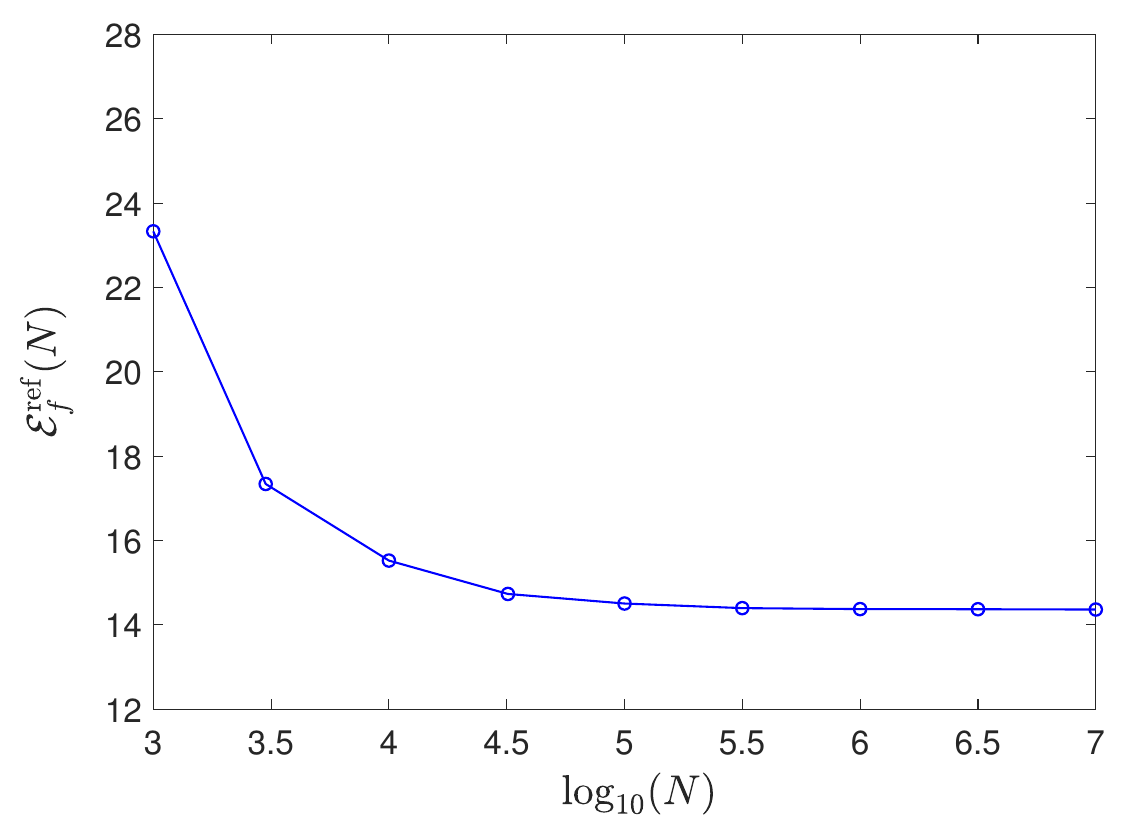}
        \caption{$N \mapsto \curE_f^\exactp (N)$ for $\nu=10$.}
        \vspace{0.3truecm}
        \label{fig:figure1d}
    \end{subfigure}
    \begin{subfigure}[b]{0.33\textwidth}
        \centering
        \includegraphics[width=\textwidth]{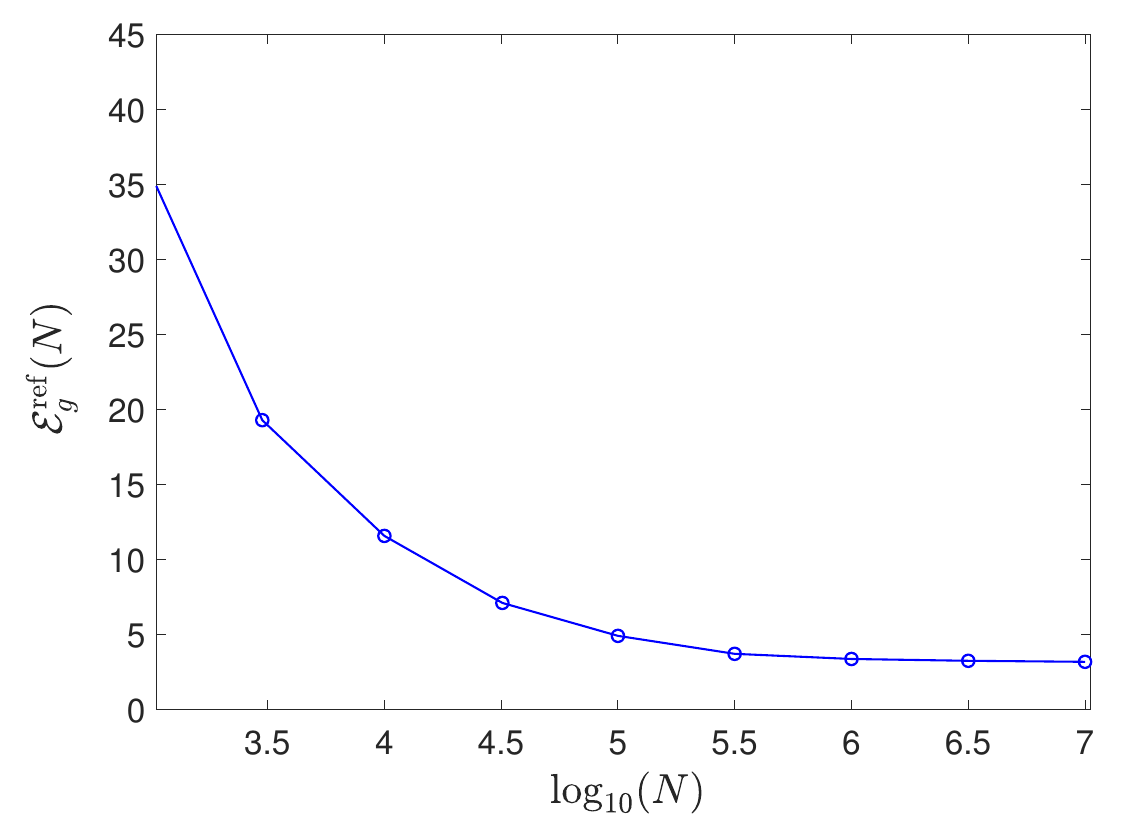}
        \caption{$N \mapsto \curE_\bfg^\exactp (N)$ for $\nu=10$.}
         \vspace{0.3truecm}
        \label{fig:figure1e}
    \end{subfigure}
    \begin{subfigure}[b]{0.33\textwidth}
        \centering
        \includegraphics[width=\textwidth]{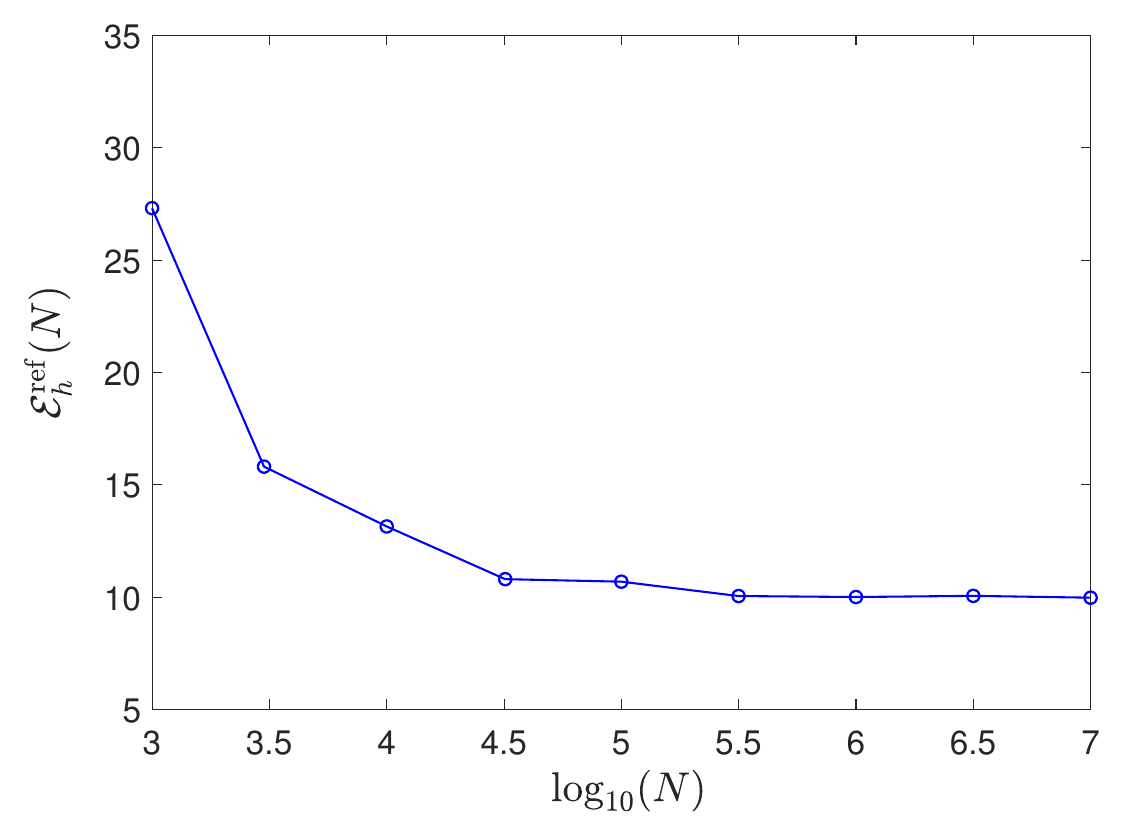}
        \caption{$N \mapsto \curE_h^\exactp (N)$ for $\nu=10$.}
        \label{fig:figure1f}
    \end{subfigure}
    \caption{Convergence analysis for the Gaussian reference case, for $\nu=1$ (a,b,c) and $\nu=10$ (d,e,f). Graph of functions
    $N \mapsto \curE_f^\exactp (N)$ (a,d), $N \mapsto \curE_\bfg^\exactp (N)$ (b,e), and $N \mapsto \curE_h^\exactp (N)$ (c,f). The horizontal axes are $\log_{10}(N)$.}
    \label{fig:figure1}
\end{figure}
\begin{figure}[h]
    \centering
    \begin{subfigure}[b]{0.3\textwidth}
    \centering
        \includegraphics[width=\textwidth]{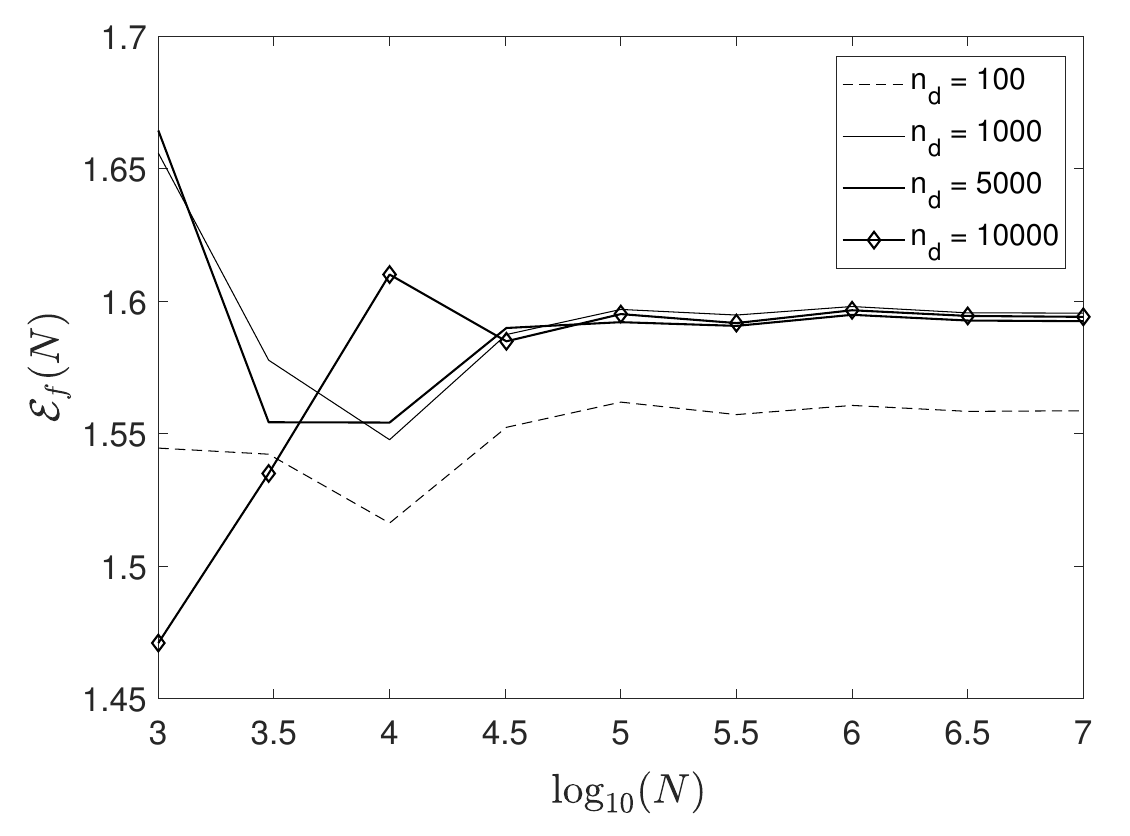}
        \caption{$N \mapsto \curE_f (N)$ for $\nu=1$.}
         \vspace{0.3truecm}
        \label{fig:figure2a}
    \end{subfigure}
    \begin{subfigure}[b]{0.3\textwidth}
        \centering
        \includegraphics[width=\textwidth]{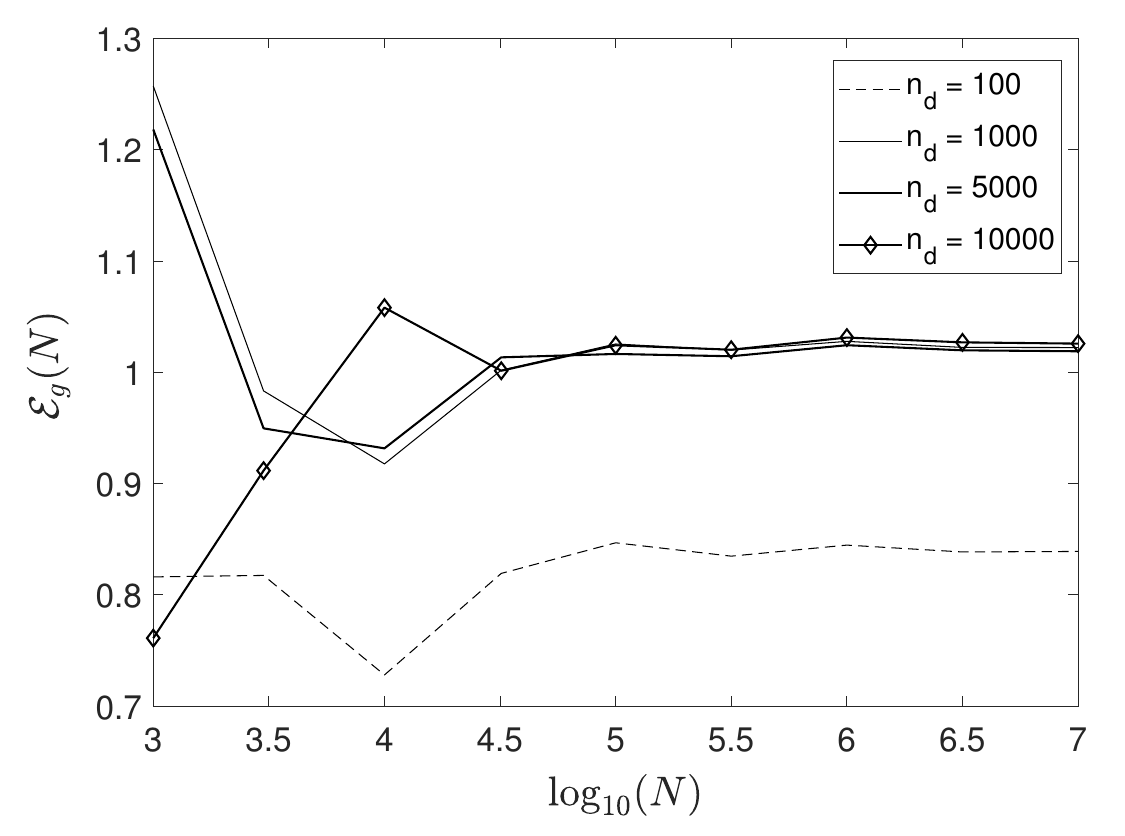}
         \caption{$N \mapsto \curE_\bfg(N)$ for $\nu=1$.}
          \vspace{0.3truecm}
        \label{fig:figure2b}
    \end{subfigure}
    \begin{subfigure}[b]{0.3\textwidth}
        \centering
        \includegraphics[width=\textwidth]{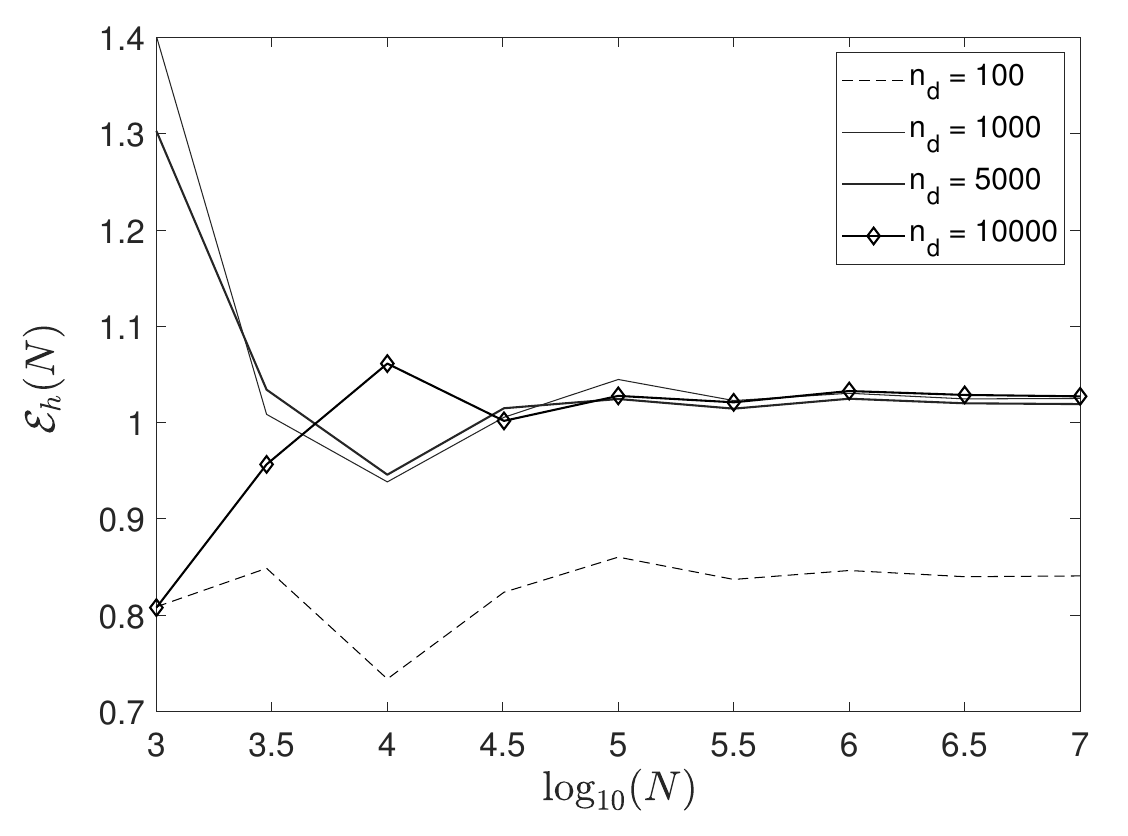}
         \caption{$N \mapsto \curE_h (N)$ for $\nu=1$.}
        \label{fig:figure2c}
    \end{subfigure}
    %
    \centering
    \begin{subfigure}[b]{0.3\textwidth}
    \centering
        \includegraphics[width=\textwidth]{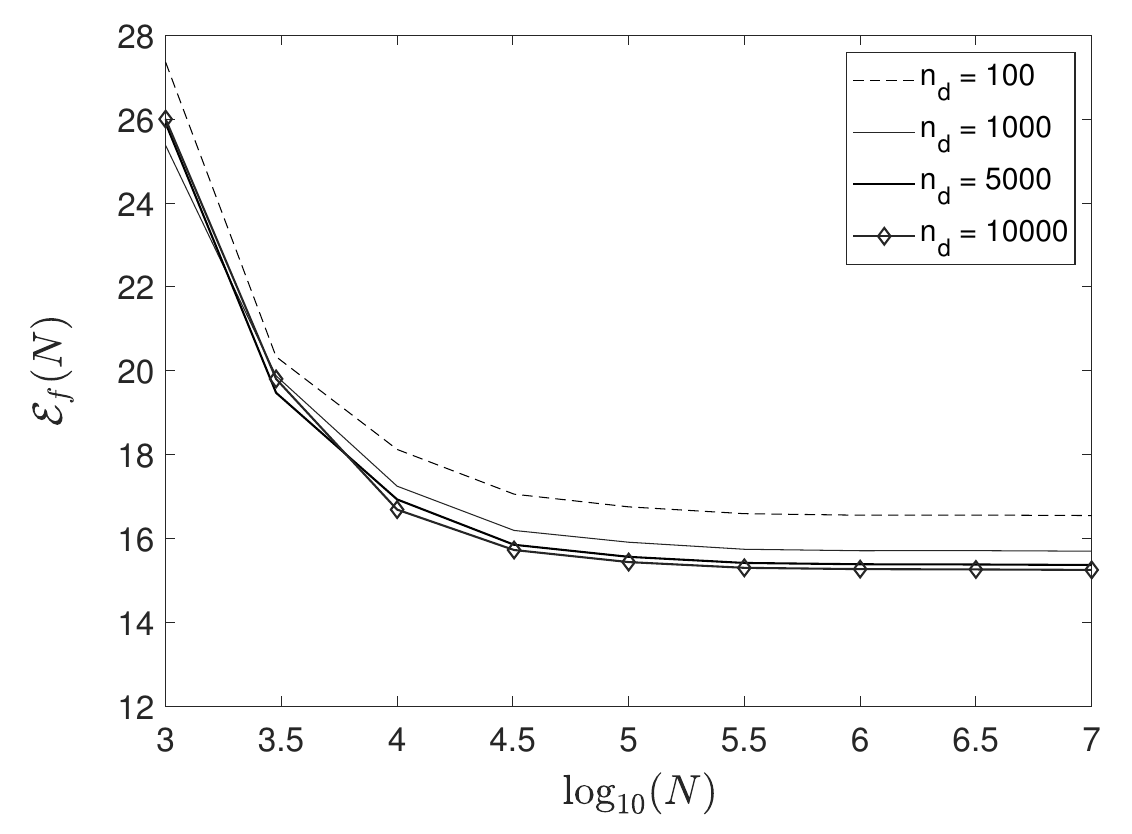}
        \caption{$N \mapsto \curE_f (N)$ for $\nu=10$.}
        \vspace{0.3truecm}
        \label{fig:figure2d}
    \end{subfigure}
    \begin{subfigure}[b]{0.3\textwidth}
        \centering
        \includegraphics[width=\textwidth]{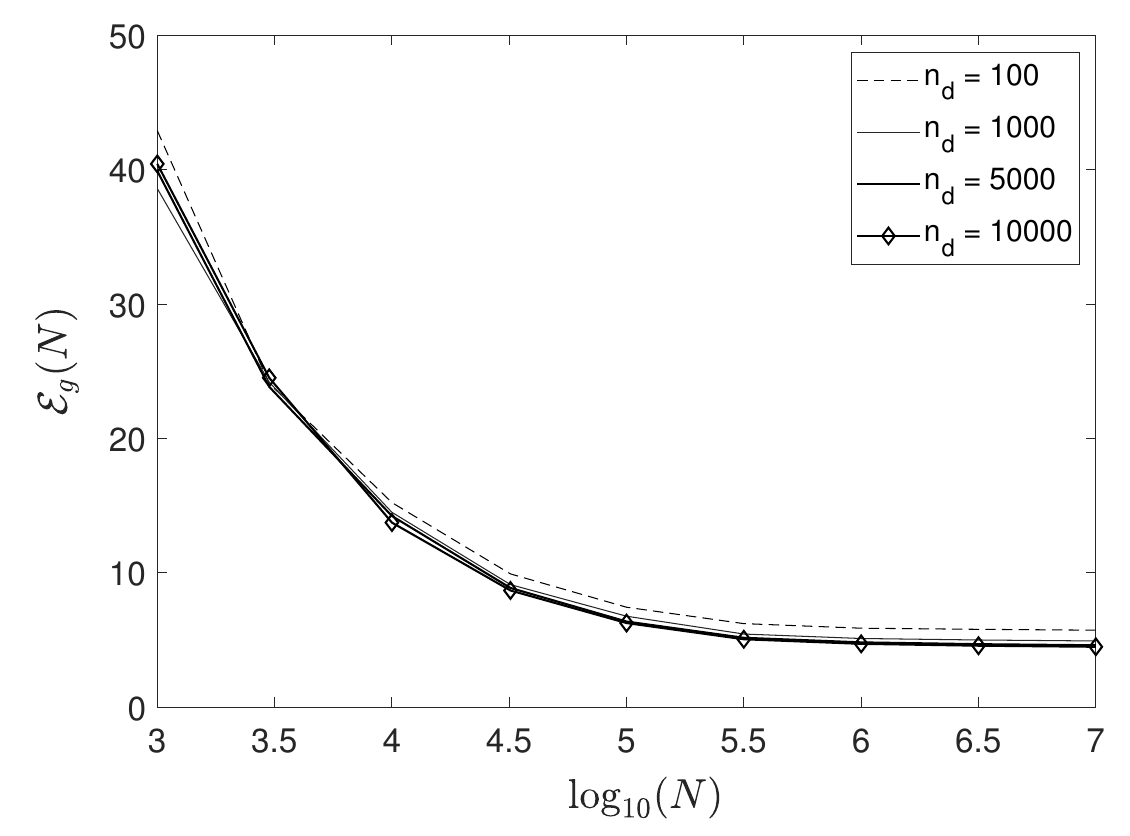}
        \caption{$N \mapsto \curE_\bfg (N)$ for $\nu=10$.}
         \vspace{0.3truecm}
        \label{fig:figure2e}
    \end{subfigure}
    \begin{subfigure}[b]{0.3\textwidth}
        \centering
        \includegraphics[width=\textwidth]{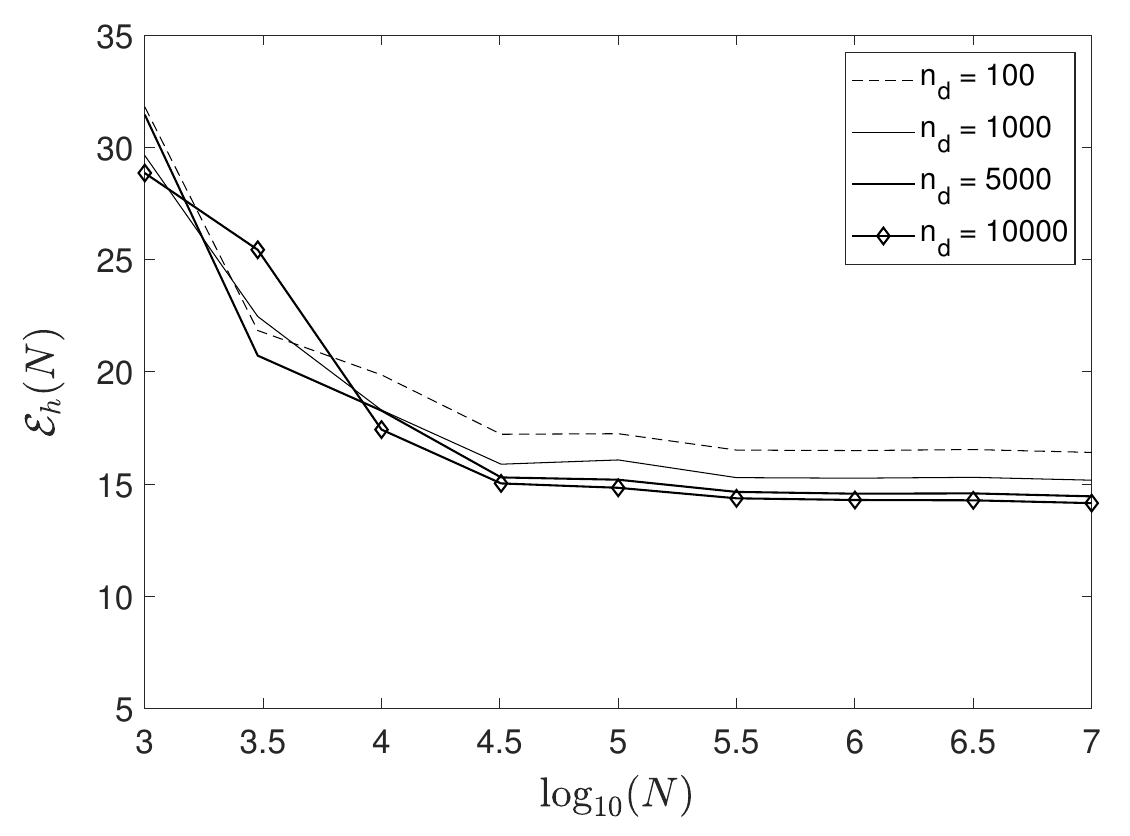}
        \caption{$N \mapsto \curE_f (N)$ for $\nu=10$.}
        \label{fig:figure2f}
    \end{subfigure}
    \caption{Convergence analysis for the Gaussian GKDE case, for $\nu=1$ (a,b,c) and $\nu=10$ (d,e,f). For $n_d = 100$ (dashed line),
    $1000$ (thin line), $5000$ (thick line), and $10\,000$ (line with diamonds), graph of functions
    $N \mapsto \curE_f (N)$ (a,d), $N \mapsto \curE_\bfg (N)$ (b,e), and $N \mapsto \curE_h (N)$ (c,f). The horizontal axes are $\log_{10}(N)$.}
    \label{fig:figure2}
\end{figure}
%
\section{Quantum computing formulation}
\label{Section6}
%
The eigenproblem given in Eq.~\eqref{eq3.16} is fully equivalent to a time-independent Schr\"odinger equation, where $\hat L_\FKP$ acts as the Hamiltonian and $q$ represents the corresponding eigenstate (or wavefunction).
Solving this eigenproblem is commonly approached in the quantum mechanics community via projection (using the variational principle) onto a tensor-product basis spanning the full space $\RR^\nu$~\cite{Gottlieb1977, Light2000, Tannor2007, Gatti2017}.
This approach is feasible only for low dimensionality $\nu$, as its computational cost scales exponentially with $\nu$.
For large dimensional problems, alternative approximations using specific tensor decompositions are sometimes employed~\cite{meyer2009, Schollwock2011, Bachmayr2023}, though these schemes do not completely alleviate the exponential scaling.
Recent advances in quantum computing offer potential for encoding the wavefunctions $q$ and their operations, even in high-dimensional problems (see for instance that quantum computers now have more than 1000 qubits~\cite{Castelvecchi2023}).
This encoding becomes possible by mapping the system, as defined by the operator $\hat L_\FKP$, onto a quantum computer, thereby representing wavefunctions and operations as qubits.
While much of the research community is focused on electronic structure problems~\cite{Fauseweh2024}, new efforts are emerging to adapt other quantum problems for quantum computation~\cite{McArdle2019,Sawaya2019,Wang2023}.
For completeness, we depict here the general approach as follows: (i) choice of a basis, (ii) second quantization, and (iii) mapping onto qubits.
%
\subsection{Choice of the basis}
\label{Section6.1}
The first step is to adopt a finite basis representation (spectral method), for which we can easily rewrite Eq.~\eqref{eq3.16} in a second quantized form.
In our case, all dimensions are distinguishable, such that operators acting on different coordinates $y_j$ commute, and we must use different basis functions for each dimension.
In line with section~\ref{Section3.7}, we employ a basis of Hermite polynomials as a working basis.
A common practice in the quantum mechanical community is to absorb the measure $p_\bfY(\bfy)\, d\bfy = \otimes_{j=1}^\nu p_Y(y_j)\, dy_j$ into the basis, where $p_Y$ is defined in Definition~\ref{definition:2}.
We also switch to the physicists' convention for the Hermite polynomials during the process of building the basis.
Furthermore, we perform an affine transformation on each coordinate, $y_j\to \sqrt{2\omega_j}(y_j-y_{j,0})$, to adapt the basis to improve convergence, where $y_{j,0}$ allows for a shift of the coordinate $y_j$, if needed.
	For coordinate $y_j$, the resulting basis $\{\varphi^{(j)}_n(y_j), n \in \mathbb{N}\}$ is such that
\begin{equation*}
\varphi^{(j)}_n(y_j) = \sqrt{2^n} \,\psi_n \Bigl(\!\sqrt{\scriptstyle 2\,\omega_j}\,(y_j-y_{j,0})\Bigr )\, \Bigl\{p_Y \Bigl (\!\sqrt{\scriptstyle 2\,\omega_j}\,(y_j-y_{j,0})\, \Bigr )\Bigr \}^{1/2}\, ,
\end{equation*}
for which the measure is $dy_j$ and where $\psi_n$ is defined in Definition~\ref{definition:2}.
These basis functions are the normalized eigenfunctions of the operator  $\omega_j \, (y_j-y_{j,0})^2-(1/\omega_j)\, \partial^2/\partial y_j^2$.
By projecting Eq.~\eqref{eq3.16} onto this basis, the eigenproblem becomes a matrix diagonalization, where the eigenstate $q_n$ is expanded as
\begin{equation}
q_n (\bfy) \,  =  \sum_{\bfalpha\in M} c_{n,\bfalpha}\, \varphi_\bfalpha(\bfy)                     \, ,\label{eq6.1}
\end{equation}
where $M=\{\bfalpha\in\NN^\nu \, \vert \,\alpha_j \leq m_j \, , \, j = 1, 2, \dots ,\nu\}$, $c_{n,\bfalpha}\in\RR$, and $\varphi_\bfalpha(\bfy)=\prod_{j=1}^\nu\varphi^{(j)}_{\alpha_j}(y_j)$, where $\{m_j\}_{j=1}^\nu$ are the maximum index that basis functions of the degree of freedom $y_j$ can have.
Therefore, there are $1+m_j$ basis functions per degree of freedom, and the total number of one-dimensional basis functions, including all degrees of freedom, is $m_T=\sum_{j=1}^\nu (1+m_j)$.
It should be noted that $m_T$ does not represent the number of basis functions in the tensor product space, which is not explicitly used in quantum computing. This concept lies at the heart of the advantage offered by quantum computing and is reminiscent of analog computing, but in the quantum realm.
%
\subsection{Second quantization}
\label{Section6.2}
On our way to express the eigenproblem on a quantum computer, we use the second quantization formalism~\cite{Surjan1989,Bowman2021}, which is particularly suited for encoding the qubits.
In this formalism, all the basis functions can be generated from $\varphi^{(j)}_0(y)$ by employing the so-called creation operator
\begin{equation*}
a^{+}_j = \frac{1}{\sqrt{2}}  ( \!\sqrt{\omega_j}\, y_j - \frac{1}{\sqrt{\omega_j}}\, \frac{\partial}{\partial y_j})\, ,
\end{equation*}
with $a^{+}_j \varphi^{(j)}_n = \sqrt{n+1}\,\varphi^{(j)}_{n+1}$ such that
$\varphi_\bfalpha (\bfy) = \Bigl(\prod_{j=1}^\nu (a^{+}_j)^{\alpha_j} \Bigr ) \,\varphi_{\bfalpha^{(0)}} (\bfy)$ with $\bfalpha^{(0)} = (0,0,\ldots ,0)$.
By use of the creation operator, one can rewrite the solution $q_n$ in the so-called Fock space, where the basis functions are described by occupation numbers that we denote by
\begin{equation*}
\varphi_\bfalpha \equiv \vert \alpha_1 \alpha_2 \dots \alpha_\nu \rangle = \Biggl ( \prod_{j=1}^\nu (a^{+}_j)^{\alpha_j} \Biggr ) \,
\vert 0 0 \ldots 0 \rangle \, .
\end{equation*}
Similarly, we define the annihilation operator
\begin{equation*}
a^{-}_j = \frac{1}{\sqrt{2}} (\!\sqrt{\omega_j}\, y_j + \frac{1}{\sqrt{\omega_j}}\, \frac{\partial}{\partial y_j})\, ,
\end{equation*}
with $a^{-}_j \varphi_n^{(j)} = \sqrt{n}\, \varphi_{n-1}^{(j)}$.
Note that $a^{+}_j$ and $a^{-}_j$ are adjoint one to the other.
Knowing the effect of $a^{+}_j$ and $a^{-}_j$ on the basis, we can rewrite all operators involved in Eq.~\eqref{eq4.32} using the identities
\begin{equation*}
y_j = \frac{1}{\sqrt{2\omega_j}}  (a^{+}_j + a^{-}_j) \quad , \quad
\frac{\partial}{\partial y_j} = (a^{-}_j - a^{+}_j)\sqrt{\omega_j/2} \, .
\end{equation*}
Then, using the commutator $a^{-}_j a^{+}_k - a^{+}_k a^{-}_j = \delta_{jk}$, one can express $\nabla^2$ and all monomials in the so-called normal ordering (with creation operators on the left of the annihilation operators) as follows,
\begin{align}
\nabla^2 \,  &=  \sum_{j=1}^\nu \frac{\omega_j}{2}\left( (a^{+}_j)^2 + (a^{-}_j)^2 - 2 a^{+}_j a^{-}_j - 1 \right)                     \, ,\label{eq6.2}\\
\bfy^\bfalpha \,
&=  \prod_{j=1}^\nu \frac{ \alpha_j! }{ \sqrt{(2\omega_j)^{\alpha_j}} } \sum_{s=0}^{\lfloor\alpha_j/2\rfloor} \sum_{p=0}^{\alpha_j-2s} \frac{ (2s-1) !\!! }{ (\alpha_j-2s-p)! (2s)! p! } (a^{+}_j)^{\alpha_j-2s-p} (a^{-}_j)^{p}                    \, ,\label{eq6.3}
\end{align}
where $n!\!!=1\times\dots\times(n-4)\times(n-2)\times n$ is the double factorial, and $\lfloor a\rfloor$ returns the integer part of $a$.
For example, a quadratic term is given by $y_j^2=( (a^{+}_j)^2 + (a^{-}_j)^2 + 2 a^{+}_j a^{-}_j + 1 )/(2\omega_j)$.
Note that Eq.~\eqref{eq6.3} results from the expansion of the operator
$y_j^{\alpha_j} = \frac{1}{\sqrt{(2\omega_j)^{\alpha_j}}}  (a^{+}_j + a^{-}_j)^{\alpha_j}$.
The double sum and the double factorial reflect the combinatorics arising from the commutation relation
$a^{-}_j a^{+}_j - a^{+}_j a^{-}_j = 1$.
With these expressions at hand, the approximation $\hat L_\FKP^{\mu,N}$ of the Fokker-Planck operator $\hat L_\FKP$, as defined by Eq.~\eqref{eq3.18}, can be expressed in a general form for a truncation at order $\mu$ with $N$ realizations of the Gaussian random variable $\bfY$,
\begin{equation}
\hat L_\FKP^{\mu,N} \, =  \sum_{\bfalpha\,\in\NN^\nu \atop |\bfalpha| \,\leq\,  2(\mu-1)} \sum_{\bfk\,\in\NN^\nu \atop \bfalpha^{(0)} \, \leq \, \bfk\,  \leq \, \,\bfalpha} G_{\bfalpha,\bfk}^N
\prod_{j=1}^\nu (a^{+}_j)^{\alpha_j} (a^{-}_j)^{k_j-\alpha_j'}                                    \, ,\label{eq6.4}
\end{equation}
where the coefficients $G_{\bfalpha\bfk}^N$ are obtained by substituting the expression from Eq.~\eqref{eq6.3} into Eq.~\eqref{eq4.32} and combining it with Eq.~\eqref{eq6.2} into Eq.~\eqref{eq3.18}.
Employing the form given in Eq.~\eqref{eq6.4} is convenient, as the action of the creation and annihilation operators can easily be computed using the fact that
\begin{equation*}
(a^{+}_j)^{\,\beta} (a^{-}_j)^\alpha\,  \vert  \varphi^{(j)}_n \rangle = \frac{\sqrt{(n+\beta-\alpha)!n!}}{(n-\alpha)!}
\, \vert \varphi^{(j)}_{n+\beta-\alpha} \rangle\, .
\end{equation*}
%
\subsection{Map onto a system of qbits}
\label{Section6.3}
Qubits are 2-level quantum systems, expressed as $\{\, c_0|0 \rangle + c_1 |1 \rangle \, , \,  \mathbf{c} = (c_0,c_1)\in \CC^2\}$, that compose the quantum computers.
To exploit quantum computing, we need to map the system of interest to the system of qubits.
Rewriting the Fokker-Planck operator in a basis of qubits is not unique.
Two mappings exist: the ``direct'' map~\cite{Somma2002}, which we employ here due to its simplicity, and the ``compact'' map~\cite{Veis2016}, which reduces the number of required qubits.
An alternative approach involves using quantum computers based on bosonic modes (harmonic oscillators) instead of qubits, or hybrid architectures combining oscillators and qubits, where additional mappings and other quantum circuits construction stategies can also be explored~\cite{Crane2024,Malpathak2024}.
Here, we utilize the direct map~\cite{Somma2002,McArdle2019,Wang2023}: the occupation of $|\varphi^{(j)}_n\rangle$ is represented by the occupation of a corresponding qubit state $| s_{j,k} \rangle$, such that a single qubit describes the occupation of a single basis function for a given degree of freedom, $y_j$.
In the direct map, the total number of qubit is given by $m_T$, the total number of basis functions across all degrees of freedom as defined in Section~\ref{Section6.1}.
The qubit state corresponding to $\varphi^{(j)}_n$, associated with the degree of freedom $y_j$, is rewritten, assuming that excatly one qubit is in the state $|1 \rangle$ at position $n$, while all others are in the state $|0 \rangle$:
\begin{equation}
|\varphi_n^{(j)}\rangle
\;\equiv\;
\underbrace{\,|(s_{j,0}=0)\rangle \,\otimes\, \cdots \,\otimes\, |(s_{j,n-1}=0)\rangle\,}_{\text{these qubits in state } |0\rangle}
\;\otimes\; |(s_{j,n}=1)\rangle \;\otimes\;
\underbrace{\,|(s_{j,n+1}=0)\rangle \,\otimes\, \cdots \,\otimes\, |(s_{j,m_j}=0)\rangle\,}_{\text{remaining qubits in state } |0\rangle}
\, . \label{eq6.5}
\end{equation}
By employing this notation, products of creation and annihilation operators on a \( n^{\text{th}} \) basis function $|\varphi_n^{(j)}\rangle$ take the form
\begin{align}
[(a_j^{+})^{\,\beta} (a_j^{-})^\alpha]_n \equiv\, &
\frac{1}{4} \frac{\sqrt{(n+\beta-\alpha)! n!}}{(n-\alpha)!}
( \hat{X}_{j,n} + i \hat{Y}_{j,n} ) \otimes ( \hat{X}_{j,n+\beta-\alpha} - i \hat{Y}_{j,n+\beta-\alpha} ) \, , \label{eq6.6}
\end{align}
where $i$ is the imaginary number, $i^2=-1$.  The operators $\hat{X}_{j,n} = (\begin{smallmatrix} 0 & 1 \\ 1 & 0 \end{smallmatrix})$ and $\hat{Y}_{j,n} = (\begin{smallmatrix} 0 & -i \\ i & 0 \end{smallmatrix})$ are Pauli matrices associated with the $k^{\text{th}}$ qubit used in the description of the $j^{\text{th}}$ degree of freedom $y_j$.
The hat indicates the matrix representation in the basis where individual qubits are represented as vectors $|0\rangle\equiv(\begin{smallmatrix}1\\0\end{smallmatrix})$ and $|1\rangle\equiv(\begin{smallmatrix}0\\1\end{smallmatrix})$.
Then, we can build the operator $(a_j^{+})^{\,\beta} (a_j^{-})^\alpha$ as
\begin{align}
(a_j^{+})^{\,\beta} (a_j^{-})^\alpha \equiv \,&
\frac{1}{4} \sum_{n=\alpha}^{(2m_j+\alpha-\beta+|\alpha-\beta|)/{2}} \frac{\sqrt{(n+\beta-\alpha)! n!}}{(n-\alpha)!}
( \hat{X}_{j,n} + i \hat{Y}_{j,n} ) \otimes ( \hat{X}_{j,n+\beta-\alpha} - i \hat{Y}_{j,n+\beta-\alpha} ) \, . \label{eq6.7}
\end{align}
Equation~\eqref{eq6.7} sums over all valid $n$ to form the full operator in qubit space, ensuring that each $| n \rangle$ can transition to $| n +\beta - \alpha \rangle$ with the usual harmonic‐oscillator amplitude.
%
\section{Quantum operations and measurements: quantum circuit}
\label{Section7}
The construction of quantum states and the measurement of their observables are fundamental steps in the development of more advanced quantum algorithms. In this section, we provide an overview of these foundational concepts, particularly for readers who may be less familiar with quantum computing.

Manipulating qubits is achieved through quantum operations that are unitary transformations.
Hence, the eigensolution that we seek can be expressed as
$|q_n\rangle=U_q|0^{m_T}\rangle$,
where $|0^{m_T}\rangle$ represents the vacuum state of ${m_T}$ qubits (i.e., all qubits initialized in their ground state).
Because individual qubits are 2-level systems, we can employ single-qubit transformations from $U(2)$, such as
$\exp(-i\theta_x \hat{X})$, $\exp(-i\theta_y \hat{Y})$, and $\exp(-i\theta_z \hat{Z})$,
where $\hat{X}$, $\hat{Y}$, and $\hat{Z}$ are the Pauli matrices (noting that $\hat{Z} = (\begin{smallmatrix} 1 & 0 \\ 0 & -1 \end{smallmatrix})$).
Transformations acting on pairs, triplets, or more qubits can be built as tensor products of unitaries from $U(2)$,
and more complex multi-qubit transformations generally involve entanglement and interactions between qubits.
For example, the Pauli group (composed of  tensor products of Pauli operators and the identity) forms a complete basis for constructing any multi-qubit unitary transformation~\cite{Nielsen2010}.
In view of Eq.~\eqref{eq6.5}, constructing a specific state is equivalent to creating a superposition of collective qubit states by applying suitable unitary transformations to single qubits, pairs of qubits, or larger subsets of qubits.
Such transformations are implemented on quantum computers via quantum gates, which are the fundamental building blocks of quantum circuits.
Examples of commonly used quantum gates~\cite{Barenco1995,Djordjevic2012} include the Hadamard gate $H$, the phase gate $S$, and the Controlled-NOT gate for two qubits, such that
\begin{equation*}
H = \frac{1}{\sqrt{2}} \left ( \begin{matrix}1&1\\
                                     1&-1
                                \end{matrix} \right ) \quad, \quad
S = \left (\begin{matrix}1&0\\
                        0&i
     \end{matrix} \right )   \quad , \quad
\rm{CNOT} = \left ( \begin{matrix}\hat 1&\hat 0\\
                                 \hat 0&\hat X
              \end{matrix} \right ) \, ,
\end{equation*}
where $\hat 1$ is the $2\times 2$ identity and $\hat X$ is the Pauli-x matrix.
To construct eigenstates, one needs a set of universal quantum gates that can reconstruct any unitary transformation in the qubit Hilbert space can be approximated to arbitrary precision~\cite{Barenco1995,Nielsen2010}.
By successively applying these gates, one can construct quantum circuits that generate specific quantum states.

The target solutions of the FKP equation span only a subset of the qubits' Hilbert space.
This subset is restricted because certain qubit states are not valid representations of the basis states defined in Section~\ref{Section6.1}. For instance, qubit states of the form
$\dots \otimes | (s_{j,k}=1) \rangle \otimes \dots \otimes | (s_{j,l}=1) \rangle \otimes \dots$
do not correspond to any valid basis state.
This is due to the constraint that, for each degree of freedom $y_j$, only one qubit can take the value $s_{j,k}=1$ among the $m_j$ qubits assigned to describe $y_j$.
Mathematically, this constraint is expressed as
\begin{equation*}
\sum_{k=0}^{m_j}(1-\langle\hat Z_{j,k}\rangle)/2=\sum_{k=0}^{m_j} s_{j,k}=1 \quad , \quad \forall j\in\{1,\dots,\nu\} \, .
\end{equation*}
To address this issue, two possible strategies can be employed:
i) adding constraints to the optimization process to explicitly enforce the condition
    $\sum_{k=0}^{m_j} s_{j,k}=1$ for all $j$, or ii) restricting the set of unitary transformations so that they inherently respect the constraint.
This can be achieved by constructing transformations exclusively from the creation and annihilation operators, ensuring that $\hat{X}$ and $\hat{Y}$ are always employed in pairs of the form
$\hat X_{j,k}\hat Y_{j,l} - \hat Y_{j,k} \hat X_{j,l}$.

Extracting information from quantum states is accomplished through quantum measurements of physical observables of interest.
This measurement process is inherently probabilistic and results in the irreversible collapse of the quantum state.
For example, assuming that an observable $\hat{A}$ has an eigendecomposition $\hat{A} = \sum_{k=1}^{2^M} |k\rangle \, a_k \, \langle k|$, where $M$ is the number of qubits, measuring $\hat{A}$ on a quantum state $|\chi\rangle$ returns an eigenvalue $a_k$ with probability $|\langle k |\chi\rangle|^2$.
After the measurement, the quantum state collapses to the corresponding eigenstate $|k\rangle$.
As a result, many measurements (realizations) are required to statistically estimate the expectation value $\langle \hat{A} \rangle = \langle \chi | \hat{A} | \chi \rangle$.
For example, in quantum chemical applications using the Variational Quantum Eigensolver (VQE) framework, the number of measurements typically required is on the order of $10^9$~\cite{Tilly2022}.
Only observables that are diagonal in the computational (qubit) basis can be directly measured, i.e., observables built as tensor products of $\hat{1}$ and $\hat{Z}$ for each qubit.
Hence, it is necessary to decompose the observable as $\hat{A} = \sum_\beta A_\beta$, where each fragment $A_\beta$ can be efficiently transformed into a diagonal form via a suitable unitary transformation $U^m_\beta$.
This decomposition implies that the state of interest must also be transformed according to the same unitary operation before measuring each fragment.
Reconstructing the expectation value $\langle \hat{A} \rangle$ requires combining all measurement outcomes on a classical computer (see, for instance, \cite{Verteletskyi2020,Yen2021,Yen2023} for efficient measurement strategies).
An alternative approach to measurement involves using the Hadamard test~\cite{Nielsen2010}, which will be discussed in Section~\ref{Section8.2} and serves as an example of a measurement technique.

Considering the two key steps—construction and measurement of a quantum state—a quantum processing unit operates in two corresponding phases: (i) construction and (ii) measurement, as illustrated in Fig.~\ref{circuit}.
The measurement process must be repeated a sufficient number of times to ensure convergence to the expectation value $\langle \hat{A}_k \rangle$ for a given fragment.
Since each measurement irreversibly collapses the quantum state, the entire process, including the state construction, must be repeated for every realization.
Then, a classical computer combines the results from all fragments to reconstruct the total expectation value $\langle \hat{A} \rangle$.

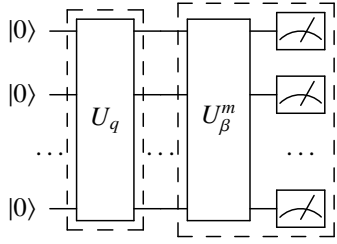
\begin{figure}
\begin{equation} \nonumber
\Qcircuit @C=1em @R=1em {
    \lstick{|0\rangle} & \multigate{3}{U_q} & \qw   & \multigate{3}{U^m_\beta} & \meter \\
    \lstick{|0\rangle} & \ghost{U_q}        & \qw   & \ghost{U^m_\beta}        & \meter \\
    \dots              & \nghost{U_q}       & \dots & \nghost{U^m_\beta}       & \dots  \\
    \lstick{|0\rangle} & \ghost{U_q}        & \qw   & \ghost{U^m_\beta}        & \meter
    \gategroup{1}{2}{4}{2}{.7em}{--}
    \gategroup{1}{4}{4}{5}{.7em}{--}
}
\end{equation}
\caption{General form of a circuit: the left dashed box indicates the trial state construction for $q_n$ in terms of qubits, while the right dashed box indicates the measurement of an observable $A_\beta$, where $U^m_\beta$ is a unitary transformation that diagonalizes $A_\beta$.}
\label{circuit}
\end{figure}

%
\section{Proposal for eigenstates construction and overlap extraction}
\label{Section8}
As explained in Section~\ref{Section3.8}, we have to compute $\bfg_{\FKP,j}^m = \vc_m(\bfeta^j)$, for all $m\in\{1,\ldots,m_\optp\}$ and for all $j\in\{1,\ldots,n_d\}$, that is to say, using Eq.~\eqref{eq3.30}, to compute $q_m(\bfeta^j)$.
This is achieved by first generating eigenstates as solution of Eq.~\eqref{eq3.23b}, and second by calculating the following overlap $q_n (\bfeta^j)  = \langle \bfeta^j | q_n \rangle$.
%
\subsection{Eigenstates construction}
\label{Section8.1}
Some of the most extensively studied algorithms for solving eigenproblems are the Variational Quantum Eigensolver (VQE)~\cite{Peruzzo2014,Tilly2022} and the Quantum Phase Estimation (QPE)~\cite{Kitaev1995,Nielsen2010}.\\

\noindent \textit{(i) QPE algorithm}. The QPE algorithm typically involves the following steps:
1) Preparing an initial quantum state with a significant overlap with the target eigenstates.
2) Decomposing the propagator $\exp(-i\hat L_\FKP t)$ into a sequence of quantum gates for various times $t$.
3) Computing and storing the overlaps between the resulting states and the initial state, also referred to as the autocorrelation function, for the chosen time intervals.
4) Performing a Fourier transform on the autocorrelation function to extract the eigenspectrum and measuring the resulting state to identify eigenstates with significant overlap with the initial state.\\

\noindent \textit{(ii) Rationale for presenting the QPE or VQE algorithm}. One objective of this section is to illustrate for readers who are not specialists in scientific quantum computing how the essential coupling between a quantum computer and a classical computer operates. Given the complexity and non-trivial nature of the QPE algorithm, we choose to focus on the VQE approach, which can be introduced using simpler concepts. Furthermore, QPE requires substantially longer quantum circuits than VQE~\cite{Reiher2017,Tilly2022}, making the latter more practical for quantum computers in the Noisy Intermediate-Scale Quantum (NISQ) era, where quantum hardware remains limited~\cite{Preskill2018}. We therefore focus on the VQE approach and provide here a brief overview.\\

\noindent \textit{(iii) VQE algorithm}. The VQE approach is a hybrid quantum-classical algorithm~\cite{McClean2016} designed to approximate the lowest, ``ground state'', eigensolutions of an eigenequation (see Fig.~\ref{circuitVQE}).
It is based on the Rayleigh-Ritz variational principle, which states that any normalized trial wavefunction $\tilde q(\bfy)$ satisfies the inequality $\curR(\tilde q)\geq\lambda_0=0$, where
\begin{equation*}
\curR(\tilde q) = \langle \hat L_\FKP \tilde q\, , \tilde q \rangle_{L^2} = \int_{\RR^\nu} \tilde q(\bfy) \, (\hat L_\FKP \tilde q(\bfy))\, d\bfy \quad , \quad \Vert \tilde q \Vert_{L^2} = \int_{\RR^\nu} \tilde q(\bfy)^2\, d\bfy =1 \, ,
\end{equation*}
is the Rayleigh quotient of the symmetric and positive operator $\hat L_\FKP$ (see Sections~\ref{Section3.5} and \ref{Section3.6}).
Equality occurs when $\tilde q(\bfy)=q_0(\bfy)$ is the ground-state wavefunction.
Algorithm VQE combines the capabilities of quantum and classical computers to iteratively refine a parameterized trial wavefunction in order to minimize $\curR(\tilde q)$.
The quantum computer is used to construct the trial wavefunction and estimate the Rayleigh quotient, while the classical computer optimizes the parameters of the wavefunction.\\

\noindent {\textit{(iii-1) Initial guess, assumed to be a good approximation to the exact solution}}.
The first step in VQE is to define an initial guess $\tilde q^{init} (\bfy)$, which serves as an approximation to the exact eigenstate.
A good initial guess reduces the number of iterations required for convergence.
The construction of $\tilde q^{init} (\bfy)$ is performed on a classical computer, using a mean-field approximation (e.g., a single Hartree product solution), or by solving a zeroth-order Hamiltonian~\cite{Wilson1980,Bowman2021}.
The initial guess is then encoded on the quantum computer using a unitary transformation $U^{init}$, such that $U^{init} |0^{m_T}\rangle=|\tilde q^{init}\rangle$.\\

\noindent {\textit{(iii-2) Parameterization of the trial wavefunction}}. The second step is parameterizing the trial wavefunction using a unitary ansatz  $U^{\VQE}(\bftau)$, where $\bftau=\{\tau_\alpha\}_{\alpha=1}^{N_{\pVQE}}$ are the $N_{\VQE}$ parameters to be optimized.
Ideally, the ansatz would explore the entire Hilbert space, for example employing operations from the Pauli group.
However, the exponential scaling of the number of elements with the number of qubits makes this computationally infeasible, as the optimization is performed on a classical computer.
Instead, a structured, parameterized ansatz, named the VQE ansatz, is employed.
Common choices include Unitary Coupled Cluster~\cite{Peruzzo2014}, Hardware-Efficient Ansatz~\cite{Kandala2017}, and Quantum Coupled Cluster~\cite{Ryabinkin2018}.
The ansatz is typically expressed as a product of parameterized unitary operations, $U^{\VQE}(\bftau) =\prod_{\alpha=1}^ {N_{\pVQE}}U^{\VQE}_\alpha(\tau_\alpha)$, and the trial state is now expressed as
$|\tilde q \rangle = \Bigl (\prod_{\alpha=1}^{N_\pVQE}U^{\VQE}_\alpha(\tau_\alpha) \Bigr ) \, U^{init}|0^{m_T} \rangle$.\\

\noindent {\textit{(iii-3) Optimization employing the classical computer}}. Once the ansatz is defined, the parameters $\bftau$ are optimized to minimize the Rayleigh quotient $\curR(\tilde q)$, which corresponds to the expectation value of $\hat L_\FKP$.
The optimization process is carried out on a classical computer using gradient-free methods, such as the Nelder-Mead simplex algorithm~\cite{Mihalikova2022}, or gradient-based methods like gradient descent. Gradients can be computed either analytically or through finite-difference approximations (see, for instance, \cite{Romero2018,Wang2018b,Bonet-Monroig2023}).
These methods require additional quantum measurements to compute extra points or derivatives, which increases the computational overhead.\\

\noindent {\textit{(iii-4) Computing other eigenstates}}. Generating other eigenstates $q_k(\bfy)$ with $k>0$, the ``excited states'', is possible within the VQE framework. Assuming that $q_0(\bfy)$ has already been obtained, the algorithm can be modified to target excited states by either adding a penalty or imposing a constraint on the symmetry of the excited state, when applicable~\cite{Ryabinkin2019}, or by enforcing the orthogonality requirement $\langle\tilde q|q_0\rangle=0$~\cite{Higgott2019}.
\begin{figure}
\centering
\includegraphics[width=0.8\textwidth]{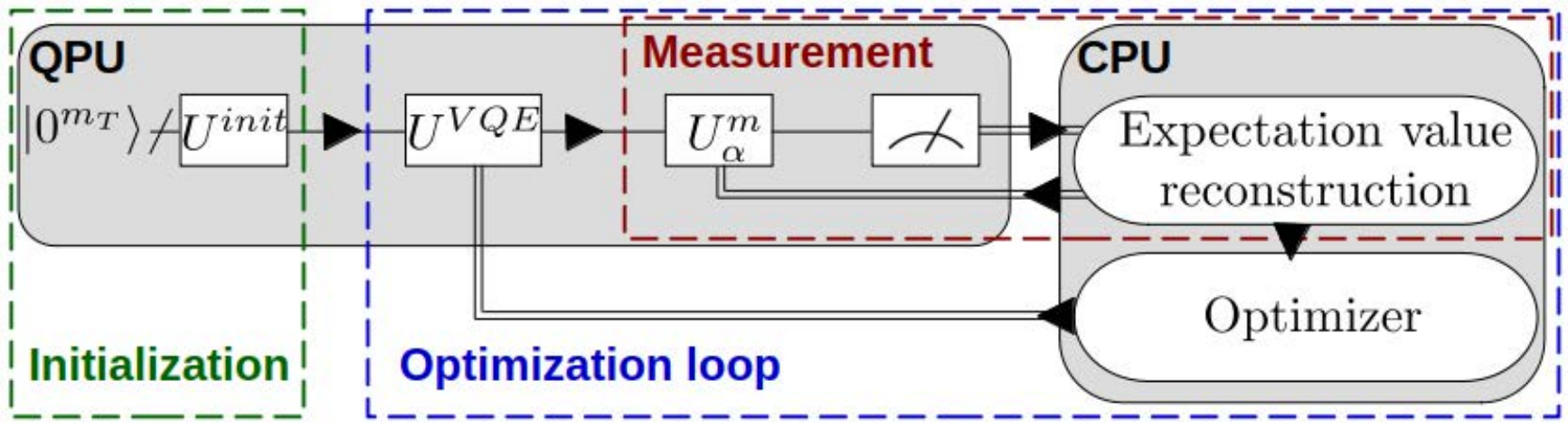}
\caption{Schematic view of a Variational Quantum Eigensolver (VQE) hybrid quantum classical circuit.
The slash `/' indicates a multi-qubit line.
The quantum state is constructed on the quantum circuit in two steps: i) the initial guess using $U^{init}$, and ii) the optimized unitary transformation $U^{\VQE}$, which depend on a set of parameters $\bftau$.
The classical computer drives the measurements, with unitaries $\{U^m_\alpha\}$ to reconstruct the expectation value of the Rayleigh quotient and its gradients with respect to the parameters $\bftau$.
The classical computer then operates a feedback on the quantum computer to optimize the trial state by changing $U^{\VQE}$.}
\label{circuitVQE}
\end{figure}
%
\subsection{Overlap extraction}
\label{Section8.2}
Assuming that eigenstates are obtained, it becomes possible to extract overlaps of the form $\langle \bfeta | q_n \rangle$.
To achieve this, we first express the overlap as $\langle \bfeta^{\parallel} | q_n \rangle$, where $|\bfeta^{\parallel}\rangle = \sum_\bfalpha \varphi_\bfalpha (\bfeta) | \varphi_\bfalpha \rangle$.
Note that $\varphi_\bfalpha$ is the tensor-product basis; consequently, this expression cannot be computed directly. This is why we rewrite it as a manageable single product.
Then, the objectives are twofold:
i) construct $| q_n \rangle$ and $|\bfeta^{\parallel}\rangle$ on the quantum computer, and
ii) extract the overlap as a measurement process.

The construction of $| q_n \rangle$ depicted in Sec.~\ref{Section8.1}.
Regarding the construction of $|\bfeta^{\parallel}\rangle$, we rewrite this state as a direct product:
\begin{equation}
|\bfeta^{\parallel}\rangle = N(\bfeta)\prod_{k=1}^\nu \left( \sum_{\alpha=0}^{m_k} C^{(k)}_\alpha (\eta_k) | \varphi^{(k)}_\alpha \rangle \right) \, , \label{eq8.1}
\end{equation}
where $N(\bfeta)$ and $C^{(k)}_\alpha (\eta_k)$ are written as
\begin{equation}
N(\bfeta) = \prod_{k=1}^\nu \left( \sum_{\alpha=0}^{m_k} |\varphi^{(k)}_\alpha (\eta_k) |^2 \right)
\quad , \quad
C^{(k)}_\alpha (\eta_k) = \frac{\varphi^{(k)}_\alpha (\eta_k)}{\sum_{\alpha=0}^{m_k} | \varphi^{(k)}_\alpha (\eta_k) |^2}
\, . \label{eq8.2}
\end{equation}
Note that substituting Eq.~\eqref{eq8.2} into \eqref{eq8.1} effectively yields $|\bfeta^{\parallel}\rangle = \prod_{k=1}^\nu \left( \sum_{\alpha=0}^{m_k} \varphi^{(k)}_\alpha (\eta_k)\, | \varphi^{(k)}_\alpha \rangle \right)$.
Each factor in the product on the right-hand-side of Eq.~\eqref{eq8.1} is constructed on the quantum computer by using a product of unitaries.
We give here a simple algorithm for their constructions.
Starting with all qubits associated with $y_k$ in their ground state, the vacuum state, we apply successive rotations to distribute the coefficients $\{C^{(k)}_\alpha\}_{\alpha=0}^{m_k}$ on the qubit states that encode $\{|\varphi^{(k)}_\alpha\rangle\}_{\alpha=0}^{m_k}$.
These rotations are defined by
\begin{equation}
\hat R_{k,\alpha} = \exp\left( -i \theta_{k,\alpha} \hat{Y}_{k,\alpha} \right)
\quad , \quad
\theta_{k,\alpha} = \mathrm{atan}\left( { C^{(k)}_\alpha }{ \left[ 1 - {\sum}_{\beta=1}^\alpha \left( C^{(k)}_\beta \right)^2 \right]^{-1/2} } \, \right) \, .\label{eq8.3}
\end{equation}
As expected, the first rotation applied on the vacuum gives a superposition of two terms: $[ 1 - ( C^{(k)}_0 )^2 ]^{1/2} | 0^{m_k} \rangle + C^{(k)}_0 | \varphi^{(k)}_0 \rangle$, the second term being the one of interest.
If we apply the second rotation on this last obtained superposition, it would destroy $C^{(k)}_0 | \varphi^{(k)}_0 \rangle$.
Instead, we must apply the second rotation only on the term containing $| 0^{m_k} \rangle$.
This is achieved by using an ``anti-controlled'' rotation,
\begin{equation*}
{\rm \bar{C}}\hat R_{k,1;0}=(\hat 1_{k,0}+\hat Z_{k,0})\otimes\hat R_{k,1}/2+(\hat 1_{k,0}-\hat Z_{k,0})/2 \, ,
\end{equation*}
which applies the rotation $\hat R_{k,1;0}$ only if $s_{k,0}=0$.
Indeed, with the controlled rotation, we now obtain
\begin{align}
{\rm \bar{C}}\hat R_{k,1;0} \left[ \sqrt{ 1 - \left( C^{(k)}_0 \right)^2 } | 0^{m_k} \rangle + C^{(k)}_0 | \varphi^{(k)}_0 \rangle \right]
&=
\sqrt{ 1 - \left( C^{(k)}_0 \right)^2 } \hat R_{k,1} | 0^{m_k} \rangle + C^{(k)}_0 | \varphi^{(k)}_0 \rangle \nonumber\\
&=
\sqrt{ 1 - \left( C^{(k)}_0 \right)^2 - \left( C^{(k)}_1 \right)^2 } | 0^{m_k} \rangle + C^{(k)}_1  | \varphi^{(k)}_1 \rangle + C^{(k)}_0 | \varphi^{(k)}_0 \rangle
\, .\label{eq8.4}
\end{align}
To generalize this scheme, we need to define anti-controlled rotations of the form ${\rm \bar{C}}\hat R_{k,\alpha;0,1,\cdots,\alpha-1}$, where the anti-control test is applied on all preceeding qubits, $\{| s_{k,\beta} \rangle\}_{\beta=0}^{\alpha-1}$.
These anti-controlled rotations are then used in place of the rotations $\hat R_{k,\alpha}$.
The circuits that achieve this construction of a factor is depicted in Fig.~\ref{circuitBuildFactor}.
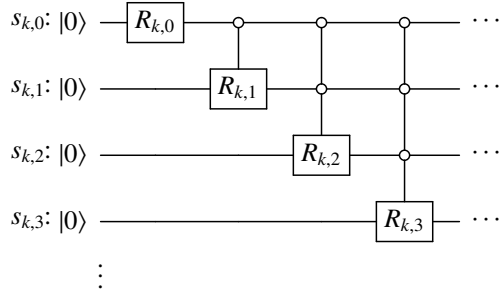
\begin{figure}
\begin{equation*}
\Qcircuit @C=1em @R=1em {
    \lstick{s_{k,0}\!\!:\,\,} & \lstick{|0\rangle} & \gate{R_{k,0}} & \ctrlo{1}      & \ctrlo{1}      & \ctrlo{1}      & \qw & \cdots \\
    \lstick{s_{k,1}\!\!:\,\,} & \lstick{|0\rangle} & \qw            & \gate{R_{k,1}} & \ctrlo{1}      & \ctrlo{1}      & \qw & \cdots \\
    \lstick{s_{k,2}\!\!:\,\,} & \lstick{|0\rangle} & \qw            & \qw            & \gate{R_{k,2}} & \ctrlo{1}      & \qw & \cdots \\
    \lstick{s_{k,3}\!\!:\,\,} & \lstick{|0\rangle} & \qw            & \qw            & \qw            & \gate{R_{k,3}} & \qw & \cdots \\
                              &          \vdots    &                &                &                &                &     &
}
\end{equation*}
\caption{Circuit for the construction of one factor in \eqref{eq8.1}, i.e. $\sum_{\alpha=1}^{m_k} \varphi^{(k)}_\alpha (\eta_k)\, | \varphi^{(k)}_\alpha \rangle$.
The lines connecting circles at the nodes indicates an anti-control (control on the zero value) on the corresponding qubits.}
\label{circuitBuildFactor}
\end{figure}
Equipped with this circuit, the unitary transformation,$U_\bfeta$, that builds $|\bfeta^{\parallel}\rangle / N(\bfeta)$ is given by
\begin{equation*}
U_\bfeta=\prod_{k=1}^\nu \prod_{\alpha=0}^{1+m_k} {\rm \bar{C}}\hat R_{k,\alpha;0,1,\cdots,\alpha-1}\, .
\end{equation*}
To extract the overlap $\langle \bfeta^{\parallel} | q_n \rangle$, we propose using the Hadamard test~\cite{Nielsen2010}.
In combination with the qubits used to construct $|q_n\rangle$, the Hadamard test employs an additional ``ancilla'' qubit.
Using the Hadamard gate and controlled-unitaries, this ancilla qubit is used to construct a superposition of $|q_n\rangle$ and $|\bfeta^{\parallel}\rangle$ in the entangled form (see Fig.~\ref{circuitHtest}),
\begin{equation*}
U_\bfeta^\dagger ( ( |\bfeta^{\parallel}\rangle/N(\bfeta) + |q_n\rangle ) \otimes|0\rangle + ( |\bfeta^{\parallel}\rangle/N(\bfeta) - |q_n\rangle ) \otimes|1\rangle )/2\, .
\end{equation*}
By measuring $\hat Z$ for the ancilla qubit, the result will be $-1$ with probability
$p_-=(1 + \Re\langle \bfeta^{\parallel}|q_n\rangle / N(\bfeta) )/2$,
or $+1$ with probability
$p_+=(1 - \Re\langle \bfeta^{\parallel}|q_n\rangle / N(\bfeta) )/2$, and where $\Re$ denotes the real part.
After reconstructing the probabilities $p_\pm$, the overlap is extracted as $\langle \bfeta^{\parallel}|q_n\rangle=(p_--p_+)\times N(\bfeta)$ (accounting for the fact that $\langle \bfeta^{\parallel}|q_n\rangle\in\RR$).
\begin{figure}
\begin{equation*}
\Qcircuit @C=1em @R=1em {
    \lstick{|0\rangle}       &     \qw & \gate{H} &                  \ctrl{1}                         & \gate{H} & \meter & \cw \\
    \lstick{|0^{m_T}\rangle} & {/} \qw &   \qw    & \gate{U_\bfeta^\dagger U_{\VQE}(\bftau) U_{init}} &   \qw    &   \qw  & \qw
}
\end{equation*}
\caption{Circuit for Hadamard test employed in measuring the overlap $\langle 0^{m_T} | U_\bfeta^\dagger U_{\VQE}(\bftau) U_{init} | 0^{m_T} \rangle = \langle \bfeta^{\parallel}|q_n\rangle$.
The slash `/' indicates a multi-qubit line. The lines connecting black point at the node indicates a control (control on the value equal to 1) on the corresponding qubits.}
\label{circuitHtest}
\end{figure}
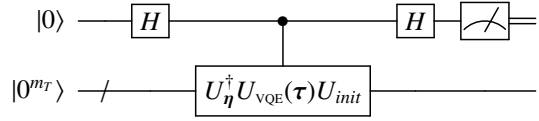
%
%
\section{Conclusions}
\label{Section9}
%
In this paper, we have developed a methodology and demonstrated a quantum computing framework to address a key challenge in probabilistic learning on manifolds (PLoM). This challenge involves solving the spectral problem of the high-dimensional FKP operator, a task that classical computational methods cannot efficiently handle. Notably, this problem may also be relevant to frameworks and applications beyond those that motivated the developments presented in this work.

In addition to the methodological aspects, we introduced a polynomial chaos expansion in the Gaussian Sobolev space for the potential in the Schr\"odinger operator associated with the FKP operator. This approach preserves the algebraic properties of the potential and enables second quantization through the introduction of creation and annihilation operators. By reformulating the eigenvalue problem in Fock space, we derived explicit formulas for both the Laplacian and the potential, facilitating the implementation of the FKP operator as a Hamiltonian in quantum circuits.

To provide a comprehensive perspective, we also included a brief overview of foundational concepts related to the construction of quantum states and the measurement of their observables, particularly for readers less familiar with quantum computing. To illustrate the essential coupling between a quantum computer and a classical computer, we chose to present the Variational Quantum Eigensolver (VQE) approach instead of the Quantum Phase Estimation (QPE) algorithm, given the latter's greater complexity and non-trivial nature.
Furthermore, QPE requires longer quantum circuits, with a larger the number of qubits than what is currently available in the NISQ era~\cite{Castelvecchi2023} and the ability for error correction~\cite{Reiher2017,Postler2022}.

The next step in this work will be the implementation of the proposed method on a quantum computer.

\section*{Conflict of interest}

The author declares that he has no conflict of interest.

%

\end{document}